\documentclass[10pt]{IEEEtran} 
\usepackage{times}
\usepackage{amsmath,amssymb,float,arydshln,color}
\usepackage{psfrag,setspace,wrapfig,subfigure,float}
\usepackage[utf8]{inputenc}
\usepackage[font=small,labelfont=bf]{caption}
\usepackage{pifont}
\usepackage{dsfont}
\usepackage{epsfig}
\usepackage{epstopdf}
\usepackage{graphicx}
\usepackage{scrextend}
\usepackage{hyperref}
\usepackage[ruled]{algorithm2e}
\usepackage{amsfonts}
\usepackage{cite}
\usepackage{hyperref}
\usepackage{url}
\usepackage{booktabs}       
\usepackage{hhline}

\usepackage[table]{xcolor}
\usepackage{multirow}
\usepackage{tabu}
\allowdisplaybreaks

\newtheorem{lemma}{{Lemma}}
\newtheorem{assumption}{{ Assumption}}
\newtheorem{theorem}{{Theorem}}
\newtheorem{corollary}{{Corollary}}
\newcommand{\HRule}{\rule{\linewidth}{0.5mm}}

\def\tran{^{\mathsf{T}}}

\def\one{\mathds{1}}

\newcommand{\bp}{ \begin{proof}}
	\newcommand{\ep}{\end{proof} }

\newcommand{\Ex}{\mathbb{E}\hspace{0.05cm}}

\newcommand{\be}{\begin{equation}}
\newcommand{\ee}{\end{equation}}
\newcommand{\bqq}{\begin{eqnarray}}
\newcommand{\eqq}{\end{eqnarray}}
\newcommand{\bal}{\begin{align}}
\newcommand{\eal}{\end{align}}
\newcommand{\bqn}{\begin{eqnarray*}}
	\newcommand{\eqn}{\end{eqnarray*}}
\newcommand{\nn}{\nonumber}
\newcommand{\ba}{\left[ \begin{array}}
	\newcommand{\ea}{\\ \end{array} \right]}
\newcommand{\qd}{\hfill{$\blacksquare$}}
\newcommand{\define}{\;\stackrel{\Delta}{=}\;}

\def\h{{\boldsymbol{h}}}

\def\n{{\boldsymbol{n}}}

\def\u{{\boldsymbol{u}}}
\def\v{{\boldsymbol{v}}}
\def\w{{\boldsymbol{w}}}
\def\x{{\boldsymbol{x}}}

\def\z{{\boldsymbol{z}}}

\newcommand{\cE}{{\mathcal{E}}}
\newcommand{\cF}{{\boldsymbol{\mathcal{F}}}}
\newcommand{\cG}{{\mathcal{G}}}

\newcommand{\cN}{{\mathcal{N}}}

\newcommand{\su}{{\scriptstyle{\mathcal{U}}}}
\newcommand{\sv}{{\scriptstyle{\mathcal{V}}}}

\newcommand{\sz}{{\scriptstyle{\mathcal{Z}}}}

\newcommand{\tw}{\widetilde{\boldsymbol{w}}}

\newcommand{\grad}{{\nabla}}

\def\filt{\boldsymbol{\mathcal{F}}}

\newcommand{\eq}[1]{\begin{align}#1\end{align}}
\newcommand{\beqn}{\begin{eqnarray}}
\newcommand{\eeqn}{\end{eqnarray}}

\newcommand{\nnb}{\nonumber \\}
\newcommand{\RR}{\mathbb{R}}

\DeclareFontFamily{U}{mathx}{\hyphenchar\font45}
\DeclareFontShape{U}{mathx}{m}{n}{
	<5> <6> <7> <8> <9> <10>
	<10.95> <12> <14.4> <17.28> <20.74> <24.88>
	mathx10
}{}
\DeclareSymbolFont{mathx}{U}{mathx}{m}{n}
\DeclareFontSubstitution{U}{mathx}{m}{n}
\DeclareMathAccent{\widebar}{0}{mathx}{"73}

\def\real{{\mathbb{R}}}

\def\Zint{{\mathchoice{\setbox1=\hbox{\sf Z}\copy1\kern-.75\wd1\box1}
		{\setbox1=\hbox{\sf Z}\copy1\kern-.75\wd1\box1}
		{\setbox1=\hbox{\scriptsize\sf Z}\copy1\kern-.75\wd1\box1}
		{\setbox1=\hbox{\scriptsize\sf Z}\copy1\kern-.75\wd1\box1}}}
	
\title{\huge SUPERVISED LEARNING UNDER DISTRIBUTED FEATURES}

\author{\IEEEauthorblockN{\it Bicheng Ying${}^{*\dagger}$, Kun Yuan${}^{*\dagger}$, and Ali H. Sayed${}^{\dagger}$\thanks{A short version of this work appears in the conference publication\cite{ying2018exponentially}. This work was supported in part by NSF grant CCF-1011918. Email: \{ybc, kunyuan\}@ucla.edu and ali.sayed@epfl.ch}}\\
	\vspace{0.3cm}
	\IEEEauthorblockA{\rm ${}^*$Department of Electrical Engineering,
	University of California, Los Angeles\\ 
	${}^\dagger$School of Engineering, 
	\'Ecole Polytechnique F\'ed\'erale de Lausanne, Switzerland}\vspace{-2mm}
	}
\begin{document}
\def\helvetica{phvr7t.tfm}
\def\helveticaoblique{phvro7t.tfm}
\def\helveticabold{phvb7t.tfm}
\def\helveticaboldoblique{phvbo7t.tfm}

\font\sfb=\helveticabold
=\helveticaboldoblique
\maketitle
\begin{abstract}\vspace{-0mm}
	This work studies the problem of learning under both large datasets and large-dimensional feature space scenarios. The feature information is assumed to be spread across agents in a network, where each agent observes some of the features. Through local cooperation, the agents are supposed to interact with each other to solve an inference problem and converge towards the global minimizer {\color{black} of an empirical risk.} We study this problem exclusively in the primal domain, and propose new and effective distributed solutions with guaranteed convergence to the minimizer {\color{black}with linear rate under strong convexity}. This is achieved by combining a dynamic diffusion construction, a pipeline strategy, and variance-reduced techniques. Simulation results illustrate the conclusions. \vspace{-1mm}
\end{abstract}
\begin{keywords}
	distributed features, dynamic diffusion, consensus, pipeline strategy, variance-reduced method, distributed optimization, primal solution.\vspace{-0mm}
\end{keywords}
\setlength{\abovedisplayskip}{1.2mm}
\setlength{\belowdisplayskip}{1.2mm}
\def\arraystretch{0.9}

\section{Introduction and Problem Formulation}\vspace{0.5mm}
{\color{black}
 Large-scale optimization problems are common in data-intensive machine learning problems \cite{bertsekas1999nonlinear,hastie2009elements, bottou2008tradeoffs,boyd2004convex}.  For applications, where both the size of the dataset and the dimension of the feature space are large, it is not uncommon for the dataset to be too large to be stored or even processed effectively at a single location or by a single agent. In this article, we examine the situation where the feature data is split across agents either due to privacy considerations or because they are already physically collected in a distributed manner by means of a networked architecture and aggregation of the data at a central location entails unreasonable costs.  More specifically, the entries (blocks) of the feature vector are assumed to be  distributed over a collection of $K$ networked agents, as illustrated in Fig.~\ref{fig-network}.
 For instance, in sensor network applications, \cite{rabbat2004distributed,kar2009distributed}, multiple sensors are normally employed to monitor an environment; the sensors are distributed over space and can be used to collect different measurements. Likewise, in multi-view learning problems\cite{sun2013survey, xu2013survey}, the observed model is represented by multiple feature sets. Another example is the Cournot competition problem in networked markets \cite{bimpikis2014cournot, metzler2003nash, yu2017distributed},  where individual factories have information about their local markets, which  will not share with each other.
 Distributed dictionary learning problems\cite{chen2015dictionary} also fit into this scenario if viewing the dictionary as feature.
}
\begin{figure}[t]
	\centering
	\includegraphics[scale=0.53]{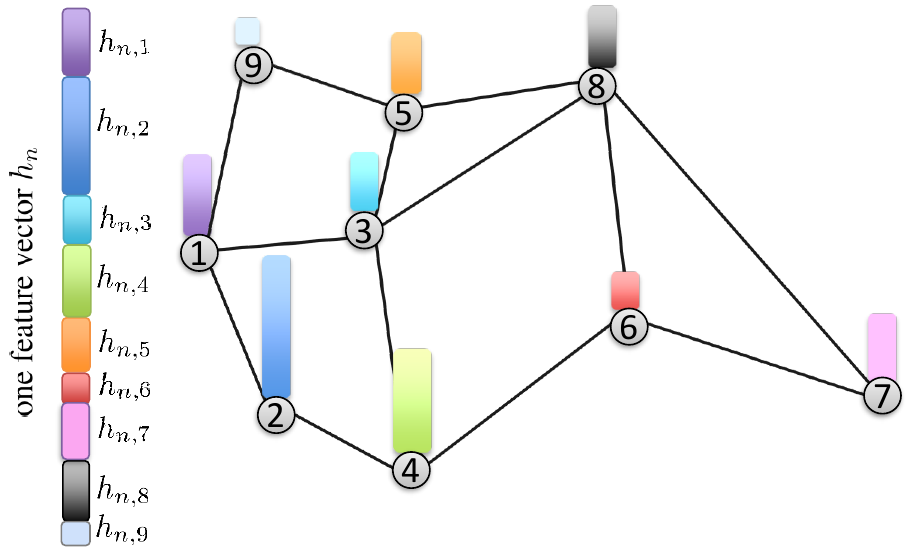}\vspace{-2mm}
	\caption{Distributing the feature across the networked agents. \label{fig-network}}\vspace{-2mm}
\end{figure}

{\color{black} In this work, we focus on empirical risk minimization problems, which, by ergodicity arguments, provide a good approximation for average risks \cite{sayed2014adaptation,  kar2009distributed,bertsekas1989parallel,duchi2011adaptive, mokhtari2016dsa, yuan2017exact1,nedic2009distributed, chen2013distributed}.  Formally, we consider an empirical risk of the form:}\vspace{-1mm}
\be
R(w) = \frac{1}{N}\sum_{n=1}^N Q\Big(h_n\tran w; \gamma_n \Big) + r(w)
\label{cost.full}
\ee
 where the unknown parameter model (or separating hyperplane) is designated by $w\in\real^{M\times1}$, while $h_n\in\real^{M\times1}$ denotes the $n$-th feature vector and $\gamma_n$ the corresponding scalar label. Moreover, the notation $Q(h\tran w;\gamma)$ refers to the loss function and is assumed to be a differentiable and convex function over $w$. In most problems of interest, the loss function is dependent on the inner product $h^{\sf T} w$ rather than the individual terms $\{h,w\}$. {\color{black} The factor $r(w)$ represents the regularization term.} We denote the minimizer of $R(w)$ in (\ref{cost.full}) by $w^{\star}$. {\color{black} Although we are assuming $w$ and $h_n$ to be column vectors, and $\gamma_n$ to be a scalar, the analysis can be easily extended to matrix quantities $W\in\real^{M\times C}$ and to vector labels $\gamma\in\real^{C\times 1}$, which we will illustrate later in the simulations. }


{\color{black} Since the entries of the feature vector are assumed distributed over $K$ agents, we partition each $h_n$, and similarly the weight vector $w$, into \ $K$ sub-vectors denoted by $\{h_{n,k}, w_k\}$}, where $k=1,2,\ldots,K$:
\be
h_n\define
\ba{c}
h_{n,1}\\\hline h_{n,2}\\[-0.2mm]\hline \vdots\\[-0.2mm] \hline h_{n,K}
\ea,\;\;\;
w\define
\ba{c}
w_1\\\hline w_2\\[-0.2mm]\hline \vdots\\[-0.2mm]\hline w_K
\ea\label{23fe.g23}
\ee
\noindent Each sub-feature vector $h_{n,k}$ and sub-vector $w_k$ are assumed  to be located at agent $k$.
{\color{black} The dimensions of $\{h_{n,k}, w_k\}$ can vary over $k$, i.e., over agents.}  In this way, the empirical risk function can be rewritten in the form 
\be  
R(w) = \frac{1}{N} \sum_{n=1}^N Q\left(\sum_{k=1}^K h_{n,k}\tran w_k; \gamma_n \right) {\color{black}+ \sum_{k=1}^K r(w_k) }
\label{eq.cost-of-sum}
\ee
{\color{black} where we are also assuming that the regularization term satisfies an additive factorization of the form
\eq{
	r(w) = \sum_{k=1}^K r(w_k) \label{eq.additive}
}
with regularization applied to each sub-vector $w_k$. This  property holds for many popular regularization choices, such as $\ell_2$, $\ell_1$, KL-divergence, etc.}
Observe that in the form \eqref{eq.cost-of-sum}, the argument of the loss function is now a sum over the inner products $h_{n,k}\tran w_k$. That is, we have a ``cost-of-sum'' form similar to what was discussed in \cite{chen2015dictionary}.\vspace{-0mm}
Our objective is to optimize (\ref{eq.cost-of-sum}) over the $\{w_k\}$ and to seek the optimal values in a distributed manner. 

\subsection{Related Works}
Problems of this type have been pursued before in the literature by using duality arguments, such as those in \cite{bertsekas1989parallel, chen2015dictionary, chang2015multi, ying2017diffusion, falsone2017dual}. {\color{black} One common way to do that is to transform problem \eqref{eq.cost-of-sum} into a constrained problem, say, as:  
	\eq{
		\min_{w\in \real^M}\;\;\;\;\;& R(w) = \frac{1}{N} \sum_{n=1}^N Q\left(z_n; \gamma_n \right) + \sum_{k=1}^K r(w_k) \label{dual.prob32}\\
		{\rm s.t.} \;\;\;\;\;\;\;& z_{n} = \sum_{k=1}^K h_{n,k}\tran w_k \;\;\;\;\; n=1,2,\ldots,N\nn
	}
{\color{black} Introducing the dual variable $y$, we get the Lagrangian function:
	\eq {
		L(y, z, w) =& \frac{1}{N} \sum_{n=1}^N Q\left(z_n; \gamma_n \right) + \frac{1}{N} \sum_{n=1}^N y_n z_{n} \nn\\
		& -\frac{1}{N} \sum_{n=1}^N y_n\sum_{k=1}^K h_{n,k}\tran w_k + \sum_{k=1}^K r(w_k) \label{23ijg/sd}
	}
}
Next, exploiting a duality argument, problem \eqref{dual.prob32} is equivalent to solving the following dual problem:
\eq{
	\min_{y \in \real^N} D(y) \!=\! \sum_{k=1}^K\left\{\!\frac{K}{N} \sum_{n=1}^N\! Q^{\star}\!\left(-y_n;\gamma_n\right) \! +  \! r^\star\!\left(\frac{1}{N}\sum_{n=1}^N y_n h_{n,k}\!\right) \!\right\}
	\label{dual.prob}
}
where the scalar $y_n$ denotes the dual variable corresponding to the $n-$th constraint, and $Q^\star(\cdot)$ and $r^\star(\cdot)$ represent the conjugate functions, i.e., $f^\star(y) = \sup_{x} (y\tran x - f(x))$, of $Q(\cdot)$ and $r(\cdot)$, respectively. Note that the function in \eqref{dual.prob} has the format of a ``sum-of-cost'' and each term inside the summation can be computed by each agent alone. Therefore, problem \eqref{dual.prob} can be solved in a number of traditional distributed algorithms \cite{yuan2017exact1,shi2015extra,nedic2009distributed,kar2011convergence,sayed2014adaptation}.
}
However, the resulting algorithms suffer from some limitations. {\color{black} One limitation is that the term $\frac{1}{N}\sum_{n=1}^N y_n h_{n,k}$ inside of $r^{\star}(x)$ has complexity $O(N)$ to compute.
}
Another limitation is that the resulting algorithms rely on the use of conjugate functions, which are not always available in closed form; this will greatly depend on the nature of the loss function $Q(\cdot)$. {\color{black} 	This limitation is worse for nonlinear models, say, for:
		\eq{
			R(w) = \frac{1}{N}\sum_{n=1}^N Q\Big(\sum_{k=1}^Kf(w_k, h_{n,k}); \gamma_n \Big) + r(w)
		}
		where $f(w_k, h_{n,k})$ is some nonlinear function. 
		These difficulties do not arise if we pursue a solution directly in the primal domain. For example, the case of nonlinear models can be handled through the chain rule of differentiation (as in backpropagation),  when deriving stochastic gradient algorithms. Furthermore, we are often interested more directly in the primal rather than the dual variable. 
		
	}
{\color{black} With regards to the large feature space, one may be motivated to consider coordinate descent techniques \cite{luo1992convergence,tseng2001convergence}, which pursue optimization one coordinate $y_n$ at a time. However, these techniques still face difficulties in both the primal and dual domains. For instance, in the primal domain \cite{mokhtari2016doubly,notarnicola2017distributed, Wang2018Coordinate}, they will generally require two time-scales: one scale governs the rate for updating the gradients and a second faster scale for running averaging iterations multiple times. This feature limits the rate at which data can be sampled because the inner calculation will need to be completed before the computation of the next datum. The same difficulties mentioned above for the dual methods will again arise if coordinate descent solutions are implemented in the dual domain \cite{Shalev-Shwartz2013,shalev2016sdca,ying2017diffusion}. For these reasons, the approach proposed in this work does not rely directly on coordinate descent implementations. 
}


Other useful approaches to solving problems of the form (\ref{eq.cost-of-sum}) are the Alternating Direction Method of Multipliers (ADMM) \cite{boyd2011distributed, chang2015multi,shi2014linear} and primal dual-methods \cite{chang2014distributed,mota2012distributed,yuan2017exact1,shi2015extra}. These techniques have good convergence properties but continue to suffer from high computational costs and two-time scale communications. {\color{black} This can be observed from \eqref{23ijg/sd} where we see that the gradient relative to the dual variable involves the term $\sum_{k=1}^K h_{n,k}\tran w_k$. This term requires extra communication. }

\subsection{Novelty and Contributions}

In this work, we propose a {\em stochastic} solution method to  (\ref{eq.cost-of-sum}) that operates directly in the {\em primal} domain. By avoiding the dual formulation, we will arrive at a simple and effective method even for large scale applications. We exploit the idea of dynamic consensus algorithm\cite{zhu2010discrete, freeman2006stability}, which has been adopted in the distributed optimization algorithms to track the average of gradients, see {\cite{di2016next,notarnicola2017distributed, nedic2017achieving,pu2018push,xin2018linear}}. Meanwhile, we are interested in tracking the sum of score, $\sum_{k=1}^K h_{n,k}\tran w_k$, due to the different problem setting.
More importantly, we will show that the proposed method is able to {\color{black} converge at a {\em linear} rate to the {\em exact} minimizer of  the empirical risk} $R(w)$ even under constant step-size learning. 
{\color{black} We will also exploit variance-reduced techniques \cite{defazio2014saga} and a pipeline strategy \cite{hennessy2011computer} to obtain a reduced complexity algorithm that requires only $O(1)$ operations per iteration.}
The algorithm will not require two (or separate) time-scales  and will not necessitate the solution of auxiliary sub-optimization problems as is  common in prior methods in the literature. 

Problems similar to \eqref{eq.cost-of-sum} were also studied in \cite{sundhar2012new} using similar primal methods like us but in a deterministic setting, but the approach is not well-suited for big-data applications.

{\color{black}
{\it Notation}: We use plain letters for
deterministic variables, and boldface letters for random variables. We also use $\Ex_x$ to denote the expectation with respect to $x$, ${\rm col}\{x_1,\cdots,x_n\}$ to denote a column vector formed by stacking $x_1,\cdots,x_n$, $(\cdot)\tran$ to denote transposition, and $\|\cdot\|$ for
the 2-norm of a matrix or the Euclidean norm of a vector.
Through the paper, we use the subscript $n$ as the index of data, the subscript $i, j$ as the index of iteration/time, and $k, \ell$ as the index of agent. We also put $i, j$ as the superscript with the same meaning, i.e.,  the index of iteration/time. The notation $\one_N= {\rm col}\{1,\ldots,1\}\in\real^N$. 
}

\section{Preliminary Na\text{\"i}ve Solution}
We first propose a simple and na\"{\i}ve solution, which will be used to motivate the more advanced algorithms in later sections. 
{\color{black}
\subsection{Networked Agent Model and Consensus Strategy}
To begin with, we introduce several preliminary concepts that will be exploited in later sections. 

We consider the graph model shown in Fig.~\ref{fig-network}. In this construction, the communication network of the agents is modeled as a fixed directed graph $\cG = (\{1,\cdots,K\}\,;\, \cE)$, where $\cE\subseteq\{1,\cdots,K\}\times \{1,\cdots,K\}$ is the set of edges.
The edge $(\ell, k)$ means that agent $\ell$ can send a message to agent $k$, where we associate the weight $a_{\ell k}$ as a nonnegative factor that scales the information from agent $\ell$ to agent $k$. 
We assume the combination matrix $A=[a_{\ell k}]$ is symmetric and doubly-stochastic, i.e.,\vspace{-1mm}
\be 
\sum_{\ell=1}^K a_{\ell k}= 1,\;\;\; \sum_{k=1}^K a_{\ell k}= 1
\ee
We also assume that  $a_{kk}>0$ for at least one agent $k$ and the underlying graph $\cG$ is strongly connected.

Now assume there is a signal $d_k$ at each agent $k$. Then, a well-studied and traditional algorithm for the agents to learn the average value of all the $\{d_k\}$ signals is to run the consensus iteration\cite{kar2011convergence,sayed2014adaptation,nedic2009distributed, boyd2006randomized}:
\eq{
	w_{i+1, k} = \sum_{\ell \in \cN_k} a_{\ell k} w_{i,\ell},\;\;\;{\rm where\ } w_{0, k} = d_k \label{static.consensus}
}
where the notation ${\cal N}_k$ denotes the set of neighbors of agent $k$. In this way, each agents starts from its own observation vector $d_k$ and continually averages the state values of its neighbors. After sufficient  iterations, it is well-known that 
\eq{
	w_{i,k}\to \frac{1}{K}\sum_{k=1}^K d_{k} 
} under some mild conditions on $A$ \cite{sayed2014adaptive,boyd2006randomized,horn1990matrix,pillai2005perron, nedic2009distributed}. 

}

\subsection{Na\text{\"i}ve Solution}
Now, let us consider the problem of minimizing (\ref{eq.cost-of-sum}) by means of a stochastic gradient recursion. Let $\alpha_{n,k}=h_{n,k}\tran w_k$ denote the inner product that is available at agent $k$ at time $n$ and define \vspace{-1mm}
\eq{\label{zn-global-sum}
	z_{n}\define\sum_{k=1}^{K}\alpha_{n,k}\\[-7mm]\nn
}
which is the argument of $Q(\cdot)$ in (\ref{eq.cost-of-sum}):
\eq{
	R(w) = \frac{1}{N} \sum_{n=1}^N Q\left(z_{n}; \gamma_n \right) + r(w)\label{eq.cost-of-sum2}
} 
\noindent If we denote the average of the  local inner products by \eq{
\bar{\alpha}_n\define \frac{1}{K}\sum_{k=1}^{K} \alpha_{n,k}}
then the variable $z_n$ is a scaled multiple of $\bar{\alpha}_n$, namely, $z_n=K\bar{\alpha}_n$. Now, the stochastic-gradient step to solving  (\ref{eq.cost-of-sum2}) will involve approximating the true gradient vector of $R(w)$ by the gradient vector of the loss function evaluated at some randomly selected data pair $(\z_{\n_i},\gamma_{\n_i})$, where $\n_i$ at iteration $i$ denotes the index of the sample pair  selected uniformly at random from the index set 
$\{1,2,\ldots, N\}$. Doing so, the stochastic gradient recursion will take the form:
\eq{\label{single-sgd}
	\w_{i+1} 
	&= \w_{i} - \mu \grad_z Q(\z_{\n_i};\gamma_{\n_i})h_{\n_i} - \mu \nabla_w r(\w_i)
}
Note that $\n_i$ is independent of the iterates $\{\w_j\}_{j=0}^i$. Recalling that $h_n$ and $w$ are partitioned into $K$ blocks, we can decompose \eqref{single-sgd} into $K$ parallel recursions run at the local agents:
\eq{\label{block-sgd}
	\w_{i+1, k} = \w_{i,k} - \mu \nabla_{z}Q\left(\z_{\n_i}; \gamma_{\n_i} \right) h_{\n_i,k}- \mu \nabla_w r(\w_{i,k})
}

One main problem with \eqref{block-sgd} is that it is {\em not} a distributed solution because agents need to calculate $\grad_z Q(z;\gamma)$ at $z$ whose value depends on all sub-vectors $\{w_k\}$ from all agents and not just on $w_k$ from agent $k$. This difficulty suggests one initial solution method. 

Since the desired variable $z_n$ is proportional to the average value $\bar{\alpha}_n$, then a consensus-type construction can be used to approximate this average. 
For some total number of iterations $J$, each  agent would start from $\bar{\alpha}_{n,k}^{(0)}=\alpha_{n,k}$ and repeat the following calculations $J$ times:
\be
\bar{\alpha}_{n,k}^{(j+1)}\;\leftarrow\;
\sum_{\ell\in{\cal N}_k}a_{\ell k} \bar{\alpha}_{n,\ell}^{(j)},\;\;\;j=0,1,\ldots, J-1 \label{consensus}
\ee
However, this mode of operation requires the agents to complete $J$ consensus updates between two data arrivals and requires a two-time scale operation: {\color{black} a faster time-scale for the consensus iterations and a slower time-scale for the data sampling and computing the gradient.} One  simplification  is to set $J=1$ and to have each agent perform only one single combination step to generate the variable:
\eq{
	\widehat{\z}_{\n_i,k} = \sum_{\ell \in \cN_k} a_{\ell k} K h_{\n_i, \ell}\tran w_{i,\ell} \label{giow.hhj}
}
where we are expressing the result of this single combination by $\widehat{\z}_{\n_i,k}$ to indicate that this is the estimate for $z_{\n_i}$ that is computed at agent $k$ at iteration $i$. 
Observe also that we are scaling the quantity inside the sum by $K$ since, as explained before, $z_n=K\bar{\alpha}_n$. 
We list the resulting algorithm in (\ref{naive.1})--(\ref{naive.3}). 

Observe that this implementation requires all agents to use the same random index $\n_i$ at iteration $i$. Although this requirement may appear restrictive, it can still be implemented in distributed architectures. For example, each agent can be set with the same random seed so that they can generate the same index variable $\n_i$ at iteration $i$. {\color{black} To agree on the same random seed in a fully distributed manner, one way is to run the consensus algorithm on the seed in the setting phase. Specifically, each agent generates a random seed number, then runs the consensus algorithm until it converges and rounds the result to the nearest integer.}
Alternatively, agents can  sample the data in a cyclic manner instead of uniform sampling. But in the main body of this paper, we assume each agent $k$ will sample the same index $\n_i$ at iteration $i$ for simplicity. \vspace{-2.5mm}

\begin{table}[t]
	\vspace{1mm}
	\noindent \HRule\\
	\noindent \textbf{\small Algorithm 1 (Na\"{\i}ve feature-distributed method for agent $k$)} \vspace{-2mm}\\[-1mm]
	\HRule\\
	{\color{black}{\bf Initialization:} \mbox{Set $w_{0,k}=0$.}}\\
	{\bf Repeat for $i=1,2,\ldots$}:\vspace{-.5mm}
	\begin{align}
	\n_i \sim \;&{\cal U}[1, N]  \;\;\;\;\;(\mbox{uniformly sampled}) \label{naive.1}\\
	\widehat{\z}_{\n_i,k} =\;& \sum_{\ell \in \cN_k} a_{\ell k} (K h_{\n_i, k}\tran w_{i,k}) \label{hat-z} \\
	\hspace{-0.5mm} \w_{i+1,k} =\;& \w_{i,k}\hspace{-1mm}-\hspace{-1mm} \mu\nabla_{z}Q\Big(\widehat{\z}_{\n_i,k}; \gamma_{\n_i} \Big)h_{\n_i,k}  -\mu\nabla_w r(\w_{i,k})
	\label{naive.3}
	\\[-5mm]\nn
	\end{align}
	\mbox{\bf End}\\[-3mm]
	\HRule\vspace{-5mm}
\end{table}

\subsection{Limitations}
Algorithm 1 is easy to  implement. However, it suffers from two major drawbacks. First, the variable $\widehat{\z}_{\n_i,k}$ generated by the combination step  \eqref{hat-z} is not generally a good approximation  for the global variable $\z_{\n_i}$. This approximation error affects the stability of the algorithm and requires the use of very small step-sizes. A second drawback is that the stochastic-gradient implementation (\ref{naive.1})--(\ref{naive.3}) will converge to a small neighborhood around the exact minimizer rather than to the exact minimizer itself \cite{yuan2016stochastic, sayed2014adaptation}.  In the following sections, we will design a more effective solution.

\section{Correction the Approximation  Error}\label{sec-consensus-correction}
\subsection{Dynamic Diffusion Strategy}
Motivated by the dynamic average consensus method \cite{zhu2010discrete, freeman2006stability}, we will design a stochastic diffusion-based algorithm to correct the error introduced by \eqref{hat-z}. {\color{black}The motivation for dynamic diffusion is similar to what been used before in the distributed optimization literature \cite{di2016next, nedic2017achieving,pu2018push}, with two main differences. First, in the current setting, we will be interested in tracking several different quantities and not a single quantity. Second, we will need to develop a stochastic (rather than deterministic) version to cope with stochasticity in inner products and/or gradients; thus resulting in an $O(1)$ operation per iteration.}

To motivate the dynamic diffusion algorithm, let us step back and assume that each agent $k$ in the network is observing some dynamic input signal, $d_{i, k} \in \RR^N$, that changes with time $i$. Assume we want to develop a scheme to ensure that each agent $k$ is able to track  the average of all local signals, i.e., $\bar{d}_i = \frac{1}{K}\sum_{k=1}^{K}d_{i,k}$. For that purpose,  we consider an optimization problem of the form:\vspace{-0mm}
\eq{
	\min_{x\in \RR^P}\; C_{i}(x) = \sum_{k=1}^K\frac{1}{2}\|x - d_{i,k}\|^2 \label{gwh89.ds}
}
where the cost function $C_i(x)$ is changing with time $i$. The global minimizer of $C_i(x)$ is the desired average $\bar{r}_i$. However, we would like the agents to attain this solution in a distributed fashion. To this end, we can apply  the exact diffusion algorithm developed in \cite{yuan2017exact1, yuan2017exact2} to solve \eqref{gwh89.ds}. For this case, the algorithm simplifies to the following recursions:
\eq{
	{\rm Adapt:}\;\;\;\;\psi_{i+1,k} =&\ x_{i,k} - \mu (x_{i,k} - d_{i+1,k}) \label{ed-adapt}\\
	{\rm Correct:}\;\;\;\;\phi_{i+1,k} =&\ \psi_{i+1,k} + x_{i,k} -  \psi_{i,k} \label{ed-correct}\\
	{\rm Combine:}\;\;\;\;x_{i+1,k} =&\ \sum_{\ell \in \cN_k}a_{\ell k} \phi_{i+1,\ell} \label{ed-combin}
}\noindent Each agent $k$ has a state variable $x_{i,k}$ at time $i$. Step (\ref{ed-adapt}) uses the input signal $d_{i+1,k}$ at agent $k$ to update its state to the intermediate value $\psi_{i+1,k}$. The second step (\ref{ed-correct}) corrects $\psi_{i+1,k}$ to $\phi_{i+1,k}$, and the third step (\ref{ed-combin}) combines the intermediate values in the neighborhood of agent $k$ to obtain the update state $x_{i+1,k}$. The process continues in this manner at all agents. Based on the results from \cite{yuan2017exact2} applied to (\ref{gwh89.ds}), we can set $\mu=1$ in \eqref{ed-adapt} and combine three recursions to get
\eq{\label{eq.dynamic-average-consensus}
	x_{i+1, k} = \sum_{\ell\in \cN_k} a_{\ell k} \left(x_{i, \ell} + d_{i+1, \ell} - d_{i, \ell}\right) 
}
with $x_{0,k} = d_{0,k}$ for any $k$. It can be shown that if the signals change slowly, the $x_{i,k}$ will track the mean $\bar{d_i}=\frac{1}{K}\sum_{k=1}^{K}d_{i,k}$ well\cite{zhu2010discrete, freeman2006stability}. Also, using induction and the initial boundary conditions, it is easy to verify that \eqref{eq.dynamic-average-consensus} has the unbiasedness property:
\eq{
	\sum_{k=1}^Kx_{i, k} = \sum_{k=1}^Kd_{i, k},\;\;\;\forall i>0\label{2gw}
} In this paper, we are interested in a second useful property that:
\eq{
	\;\;\;\lim_{i\to\infty} x_{i, k} = \frac{1}{K}\sum_{k=1}^{K} d_{k},\;\;\; {\rm when \ }	\lim_{i\to\infty} d_{i, k} = d_{k},\;\;\; \forall k 
}
This means if the signals $d_{i,k}$ converge, then the $x_{i,k}$ of all agents will converge to the mean of the limit values. We refer to \eqref{eq.dynamic-average-consensus} as the dynamic diffusion method.
%
%
%

We now apply this intermediate result to the earlier recursion (\ref{block-sgd}) to transform it into a distributed solution.  Recall that there we  need to evaluate the variable \vspace{-0mm}
\eq{\label{xcnwe8}
	\z_{\n_i} = \sum_{k=1}^{K} h_{\n_i,k}\tran \w_{i,k}
} 
Calculating this quantity is similar to solving problem \eqref{gwh89.ds}, where each $d_{i,k}$ corresponds to the inner product $h_{\n_i,k}\tran \w_{i,k}$. However, there is one key difference: the signal $h_{\n_i}$ is not {\em deterministic} but {\em stochastic} and it varies randomly with the data index $\n_i$. 
At any particular iteration $i$, we do not know beforehand which random index $\n_i$ is going to be chosen. This suggests that in principle we should keep track of $N$ variables $z_n$, one for each possible $n=1,2,\ldots,N$. {\color{black} For large datasets, this is of course too costly to implement it.} Instead, we propose a more efficient solution where the data is sparsely sampled. Assume first, for the sake of argument only, that we move ahead and compute the variable $z_n$ for every possible value of $n$. If we do so, we would need to repeat construction (\ref{eq.dynamic-average-consensus}) a total of $N$ times at each node $k$, one for each $n$, as follows:
\eq{
\hspace{-4mm}	\z^{i+1}_{1,k} &= \!\sum_{\ell \in \cN_k}a_{\ell k} \left(\z^i_{1,\ell} + K h_{1,\ell}\tran \w_{i,\ell} - Kh_{1,\ell}\tran \w_{i-1,\ell} \right)\label{238dh-1}\\
\hspace{-4mm}		\z^{i+1}_{2,k} &= \!\sum_{\ell \in \cN_k}a_{\ell k} \left(\z^i_{2,\ell} + K h_{2,\ell}\tran \w_{i,\ell} - Kh_{2,\ell}\tran \w_{i-1,\ell} \right)\\[-2mm]
	&\vdots\nn\\
\hspace{-4mm}	\z^{i+1}_{N,k} &= \!\sum_{\ell \in \cN_k}a_{\ell k} \left(\z^i_{N,\ell} + K h_{N,\ell}\tran \w_{i,\ell} - Kh_{N,\ell}\tran \w_{i-1,\ell} \right)\label{238dh-3}
}
In this description, we are adding a superscript $i$ to each $\z^{i}_{n,k}$ to indicate the iteration index. In this way, each $\z_{n,k}^i$ will be able to track the sum $\sum_{k=1}^{K}h_{n,k}\tran w_{i,k}$. However, since the data size $N$ is usually very large, 
it is too expensive to communicate and update all $\{\z_{n,k}\}_{n=1}^N$ per iteration. We  propose a stochastic algorithm in which only {\em one} datum $h_{\n_i, k}$ is selected at iteration $i$ and only the corresponding entry  $z_{\n_i,k}^{i+1}$ be updated while all other $\z_{n,k}^{i+1}$ will stay unchanged for $n\neq \n_i$: 
\eq{
	\begin{cases}
		\z^{i+1}_{\n_i,k} &\hspace{-3mm}= \displaystyle \sum_{\ell \in \cN_k}a_{\ell k} \left(\z^j_{\n_i,\ell} + K h_{\n_i,\ell}\tran \w_{i,\ell} - Kh_{\n_i,\ell}\tran \w_{j-1,\ell} \right)\\
		\z^{i+1}_{n,k} & \hspace{-3mm}= \z^{i}_{n,k},\;\;\;\;\;\;\;\;\;\;n\neq\n_i  
	\end{cases}\raisetag{14pt}\label{2389hgdn}
}
where the index $j$ in the first equation refers to the most recent iteration where the same index $\n_i$ was chosen the last time.  Note that the value $j$ depends on $\n_i$ and the history of sampling, and therefore we need to store the inner product value that is associated with it.  To fetch $\z_{n,\ell}^j$ and  $Kh_{\n_i,\ell}\tran \w_{j-1,\ell}$ easily, we introduce two auxiliary variables:
\eq{
	\u^i_{\n_i,\ell} \leftarrow\,& \z^j_{\n_i,\ell},\;\;\;\;\v^i_{\n_i,\ell} \leftarrow\, Kh_{\n_j,\ell}\tran \w_{j-1,\ell}
}
{\color{black} If we view $\{\z^{i+1}_{1,k}, \z^{i+1}_{2,k}, \ldots, \z^{i+1}_{N,k}\}$ as one long vector, the update \eqref{238dh-1}---\eqref{238dh-3} resembles a  coordinate descent algorithm\cite{tseng2001convergence,notarnicola2017distributed}. 
}

{\color{black}
An unbiasedness property similar to \eqref{2gw} will continue to hold: 
\eq{
	\sum_{k=1}^K\u^i_{n,k} = \sum_{k=1}^K\v^i_{n,k} = \sum_{k=1}^K Kh_{n,k}\tran \w_{j-1,k},\;\;\;\;\forall \, n, i>0\label{dd.unbiased}
}
It is easy to verify the validity  of \eqref{dd.unbiased} by taking the summation over \eqref{2389hgdn} and combining the initialization conditions.

}

\subsection{Variance-Reduction Algorithm}
We can enhance the algorithm further by  accounting for the gradient noise that is present in (\ref{single-sgd}): this noise is the difference between the gradient of the risk function $R(w)$, which is unavailable,  and the gradient of the loss function that is used in (\ref{single-sgd}).  It is known, e.g., from \cite{polyak1987introduction,sayed2014adaptation,yuan2016stochastic,nedic2001convergence} that under constant step-size adaptation, and due to this gradient noise, recursion \eqref{single-sgd} will only approach an $O(\mu)-$neighborhood around the global minimizer of \eqref{eq.cost-of-sum}. We can modify the recursion to ensure convergence to the exact minimizer as follows. 

There is a family of variance-reduction algorithms such as SVRG\cite{johnson2013accelerating}, SAGA\cite{defazio2014saga}, and AVRG\cite{ying2017convergence} that can approach the exact solution {\color{black} of the empirical risk function} with constant step-size. In this work, we exploit the SAGA construction because the variables $\{\u_{n,k}\}$ can readily be used in that implementation. 
{\color{black} Let us consider an agent $k$ in a non-cooperative scenario where the vanilla SAGA recursion would be 
\eq{
	\hspace{-2mm}\w_{i+1, k}& \!=\! \w_{i,k}\!-\!\mu \hspace{-0.5mm}\grad_z Q( h\tran_{\n_i,k}\w_{i,k};\gamma_{\n_i}) h_{\n_i,k} \label{eq.saga-1}\\
	&  \;\;\;\;+ \mu\grad_z Q(\u_{\n_i, k}^i;\gamma_{\n_i}) h_{\n_i, k}\nn\\
	& \;\;\;\; - \frac{\mu}{N}\sum_{n=1}^{N}\grad_z Q(\u_{n, k}^i;\gamma_{n}) h_{n, k}  \hspace{-0.5mm} \nn\\
	\u_{n,k}^{i+1} = &
	\begin{cases}
		h\tran_{\n_i,k}\w_{i,k}, & \hspace{1.mm}\mbox{if $n=\n_i$} \\
		\u_{n,k}^{i}, & \hspace{1.mm}\mbox{otherwise}
	\end{cases}
}
It is proved in \cite{defazio2014saga} that the variance of the gradient noise, i.e., the difference between the gradient step in \ref{eq.saga-1} and the full gradient, introduced in SAGA 
will vanish  in expectation. Therefore, SAGA will converge to the exact solution of problem \eqref{cost.full}.
Also, note that the $N$-summation term $\sum_{n=1}^{N}\grad_z Q(\u_{n}^i;\gamma_{n}) h_{n}$ can be calculated in an online manner since only one term is updated, i.e.,
\eq{
	&\hspace{-3mm}\sum_{n=1}^{N}\grad_z Q(\u_{n}^{i+1};\gamma_{n}) h_{n}\\
	 =\!& \sum_{n=1}^{N}\!\grad_{\!\!z}Q(\u_{n}^i;\gamma_{n}) h_{n}\!-\! \grad_{\!\!z} Q(\u_{\n_i}^i;\gamma_{\n_i}) h_{\n_i} \! +\! \grad_{\!\!z} Q(\u_{\n_i}^{i+1};\gamma_{\n_i}) h_{\n_i} \nn
}
This online calculation results  in $O(1)$ complexity per iteration. 

Note that the vanilla SAGA needs to store the gradient $\grad_z Q(\u_{n, k}^i;\gamma_{n})h_{n, k}$ for $n=1,2,\ldots,N$. However, we have already stored $\u_{n,k}^i$ and the additional storage of the gradient is unnecessary.
}

Hence, the stochastic gradient recursion(\ref{block-sgd}) at each agent $k$ will be modified to (\ref{w-algorithm3}) with two correction terms.
The resulting algorithm is summarized in Algorithm 2.
\smallskip

\begin{table}[htbp]
	\noindent\HRule\\
	\textbf{Algorithm 2 [Variance-reduced dynamic diffusion ($\rm\bf VRD^2$) for learning from distributed features]}\\[-2mm]
	\noindent\HRule\\
	{\bf Initialization:} \mbox{Set $w_{0,k}=0$; $\u_{n,k}^0=0;\ \v_{n,k}^0=0$.} \\
	{\bf Repeat for $i=1,2,\ldots$}:\vspace{-1mm}
	\begin{align}
	\n_i \sim \;&{\cal U}[1, N] \;\;\;\;\;(\mbox{uniformly sampled})\\
	\z_{\n_i,k} =& \sum_{\ell \in \cN_k}\!a_{\ell k}\! \left(\u^i_{\n_i,\ell} \!+\! K h_{\n_i,\ell}\tran \w_{i,\ell} \!-\!  \v^i_{\n_i,\ell}\right) \label{z-algorithm3}\\
	\hspace{-2mm} \w_{i+1,k} =\;& \w_{i,k} \hspace{-1mm}-\hspace{-1mm} \mu\Bigg\{\!\!\Big[\nabla_{z}Q\hspace{-0.5mm}\big(\hspace{-0.5mm}\z_{\n_{i}, k};\hspace{-0.5mm} \gamma_{\n_{i}}\hspace{-0.5mm}\big)\hspace{-1mm} -  \hspace{-1mm}\nabla_{z}Q\hspace{-0.5mm}\big(\hspace{-0.5mm}\u^i_{\n_i,k};\hspace{-0.5mm}\gamma_{\n_{i}}\hspace{-0.5mm}\big)\hspace{-0.5mm} \Big]h_{\n_i,k} \nnb
	&\hspace{1.cm} + \frac{1}{N}\sum_{n=1}^{N}\nabla_{z}Q\Big(\u^i_{n,k}; \gamma_{n}\Big)h_{n,k}+\nabla_w r(\w_{i,k})\!\Bigg\} \label{w-algorithm3}\\
	\u_{n,k}^{i+1} = &
	\begin{cases}
	\z^{i+1}_{\n_i,k}, & \hspace{1cm}\mbox{if $n=\n_i$} \\
	\u_{n,k}^{i}, & \hspace{1cm}\mbox{otherwise}
	\end{cases}\\
	\v_{n,k}^{i+1} = &
	\begin{cases}
	Kh_{\n_i,k}\tran \w_{i, k}, & \mbox{if $n=\n_i$} \\
	\v_{n,k}^{i}, & \mbox{otherwise}
	\end{cases}\\[-7mm]\nn
	\end{align}
	\mbox{\bf End}\\[-2mm]
	\noindent\HRule\vspace{-2mm}
\end{table}



\section{Acceleration With Pipeline}
 Algorithm 2 can be shown to converge to the solution of problem \eqref{cost.full} for sufficiently small step-sizes. However, it is observed in numerical experiments that its convergence rate can be slow. One reason is that the variable $\z_{n,k}$ generated by \eqref{z-algorithm3} converges slowly to $\sum_{k=1}^{K}h_{n,k}\tran \w_{i,k}$. To accelerate  convergence, it is necessary to run \eqref{z-algorithm3} multiple times before the gradient descent step \eqref{w-algorithm3}, which, however, will take us back to  in a two-time-scale algorithm. In this section, we propose a pipeline method that accelerates the convergence of  $\z_{n,k}$ while maintaining the one-time-scale structure.



%
%
%
%

A pipeline is  a set of data processing elements connected in series, where the output of one element is the input to the next element\cite{hennessy2011computer}. We assume each agent $k$ stores $J$ variables at iteration $i$:
\eq{
	[\z_{\n_i,k}^{(0)}, \z_{\n_{i-1},k}^{(1)}, \cdots \z_{\n_{i-J+1},k}^{(J-1)}] \in \RR^J\label{state.vec}
}
At every iteration, agent  $k$ runs a {\em one-step} average consensus recursion on its state vector (\ref{state.vec}):
\eq{\label{consensus-pipeline}
	&\hspace{-5mm}[\z_{\n_i,k}^{(1)}, \z_{\n_{i-1},k}^{(2)}, \cdots \z_{\n_{i-J+1},k}^{(J)}] \nn\\
	&= \sum_{\ell \in \cN_k} a_{\ell k } [\z_{\n_i,\ell}^{(0)}, \z_{\n_{i-1},\ell}^{(1)}, \cdots \z_{\n_{i-J+1},\ell}^{(J-1)}] \\[-7mm]\nn
}
Then,  agent $k$ pops up the variable $\z_{\n_{i-J+1},k}^{(J)}$ from memory and uses it to continue the stochastic gradient descent steps. Note that the variable  $\z_{n,k}^{(J)}$ can be interpreted as the result of applying $J$ consensus iterations and, therefore, it is a better approximation for  $\sum_{k=1}^{K}h_{n,k}\tran \w_{n,k}$. At iteration $i+1$, agent $k$ will push a new variable $\z_{\n_{i+1},k}^0 = K \h_{\n_{i+1}, k}\tran \w_{i+1, k}$ into the buffer
and update its state to
\eq{\label{sns8}
	[\z_{\n_{i+1},k}^{(0)}, \z_{\n_{i},k}^{(1)}, \cdots \z_{\n_{i-J+2},k}^{(J-1)}] \in \RR^J.
}\vspace{-3mm}
\begin{figure}[t]
	\centering
	\includegraphics[scale=0.5]{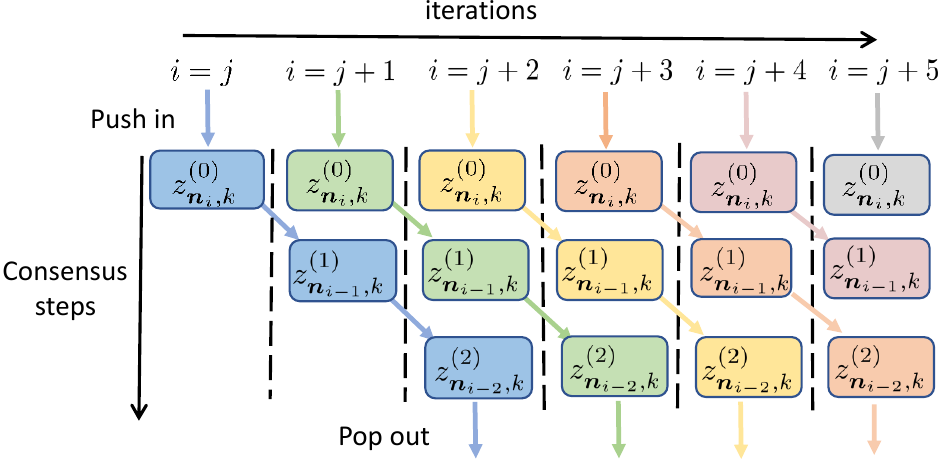}\vspace{-1mm}
	\caption{Illustration of the pipeline strategy with buffer length $J=3$.}\vspace{-1mm}\label{fig.illustration}
\end{figure}

Recursion \eqref{consensus-pipeline} employs the pipeline strategy. For example, variable $\z_{\n_i,k}^{(1)}$ is updated at iteration $i$. This new output $\z_{\n_i,k}^{(1)}$ will become  the second input at iteration $i$ and is used to produce the output $\z_{\n_i,k}^{(2)}$. Next, the output $\z_{\n_i,k}^{(2)}$ will be the third input at iteration $i+2$ and is used to produce the output $\z_{\n_i,k}^{(3)}$. If we follow this procedure, the output $z_{\n_i,k}^{(J)}$ will be reached at iteration $i+J-1$. At that time, we can pop up $z_{\n_i,k}^{(J)}$ and use it in the stochastic gradient update. The pipeline procedure is summarized in the ``Pipeline function'' shown above, and Fig.~\ref{fig.illustration} illustrates the pipeline strategy. \\[-0mm]
\begin{table}[!h]
\noindent\HRule\vspace{-0.3cm}
\begin{center}
	\bf Pipeline function
\end{center}\vspace{-0.5cm}
\noindent\HRule\\
{\bf Initialization}: \mbox{$\z_{\n_i,k}=0$ for any $i\le 0$}\\
{\bf Function } {{\rm  Pipeline}} \big($\z^{(0)}_{\n_i,k}$, $\v^{i+J-1}_{\n_i,k}$\big)
\eq{
	&\hspace{-10mm}\mbox{Push\;\; \Big[$\z^{(0)}_{\n_i,k}$, $\v^{i+J-1}_{\n_i,k}$\Big]\; into the queue} \\
	\vline&\hspace{-10mm}[\z_{\n_i,k}^{(1)}, \z_{\n_{i-1},k}^{(2)}, \cdots \z_{\n_{i-J+1},k}^{(J)}]\nn\\
	&= \sum_{\ell \in \cN_k} a_{\ell k } [\z_{\n_i,\ell}^{(0)}, \z_{\n_{i-1},\ell}^{(1)}, \cdots \z_{\n_{i-J+1},\ell}^{(J-1)}] \label{dist.comb}\\
	&\hspace{-10mm}\mbox{Pop\;\;\Big[$\z^{(J)}_{\n_{i-J+1},k}$, $\v^{i}_{\n_{i-J+1},k}$\Big] out of the queue} 
}
{\bf Return} \Big[$\z^{(J)}_{\n_{i-J+1},k}$, $\v^{i}_{\n_{i-J+1},k}$\Big] 
\\
\noindent\HRule
\end{table}
The pipeline strategy has two advantages. First, it is able to calculate $z_{n,k}^{(J)}$ without inner loop, which accelerates the algorithm and maintains the one-time-scale structure. Second, in one iteration, the two-time-scale solution sends a scalar $J$ times while the pipeline solution sends a $J$-length vector just once. Though the communication load is the same, the pipeline solution is usually faster than the two-time-scale approach in practice. That is because sending a scalar with $J$ time needs all agents to synchronize for $J$ times, which can take longer time than the one-time communication.

Observe that when $z_{\n_i, k}^{(0)}$ is popped out of the pipeline after $J$ iterations, the required $\v^i_{\n_i, k}$ should also be from $J$ iterations ago. This means that we need to store the past $\v^i_{\n_i, k}$ values. One solution is to push the auxiliary $\v^i_{\n_i, k}$ along with $z_{\n_i, k}^{(0)}$ into the pipeline but without doing any processing inside the pipeline. Since this value will be used $J$ iterations into the future, this variable should be denoted by $\v^{i+J-1}_{\n_i, k}$. \vspace{2mm}

\arraycolsep=1.4pt
\begin{table}
	\noindent\HRule\\
	\textbf{Algorithm 3 [Pipelined variance-reduced dynamic diffusion ($\rm \bf PVRD^2$) learning]}\\[-2mm]
	\noindent\HRule\\
	{\bf Initialization:} \mbox{Set $w_{0,k}=0$; $\u_{n,k}^0=0;\ \v_{n,k}^0=0$.} \\
	{\bf Repeat for $i=1,2,\ldots$}:\vspace{-2mm}
	\begin{align}
	\n_i \sim \;&{\cal U}[1, N] \;\;\;\;\;(\mbox{uniformly sampling})\\
	&\hspace{-1cm}\Big[\z^{(J)}_{\n_i',k}, \v^{i+1}_{\n_i', k}\Big] \hspace{6mm} ({\rm denote\ } \n_i' \define \n_{i-J+1})\label{pipeline.input}\\
	=\;&{\rm  Pipeline}\!\left(\!\u^i_{\n_i, k} \!+\! Kh_{\n_i,k}\tran \w_{i,k}  \!-\! \v^i_{\n_i,k},\; Kh_{\n_i,k}\tran \w_{i,k}\right)\nn\\
	\hspace{-2mm} \w_{i+1,k} =\;& \w_{i,k} \hspace{-1mm}-\hspace{-1mm} \mu\Bigg\{\!\!\Big[\nabla_{z}Q\hspace{-0.5mm}\big(\hspace{-0.5mm}\z_{\n'_{i}, k}^{(J)};\hspace{-0.5mm} \gamma_{\n'_{i}}\hspace{-0.5mm}\big)\hspace{-1mm} -  \hspace{-1mm}\nabla_{z}Q\hspace{-0.5mm}\big(\hspace{-0.5mm}\u^i_{\n'_i,k};\hspace{-0.5mm}\gamma_{\n'_{i}}\hspace{-0.5mm}\big)\hspace{-0.5mm} \Big]h_{\n'_i,k} \nnb
	&\hspace{1.1cm} + \!\frac{1}{N}\sum_{n=1}^{N}\nabla_{z}Q\Big(\u^i_{n,k}; \gamma_{n}\Big)h_{n,k}+\nabla_w r(\w_{i,k})\!\Bigg\}\label{w-algorithm4}\\
	\u_{n,k}^{i+1} = &
	\begin{cases}
	\z_{\n'_i,k}^{(J)}, & \hspace{1.mm}\mbox{if $n=\n'_i$} \\
	\u_{n,k}^{i}, & \hspace{1.mm}\mbox{otherwise}
	\end{cases}\label{update.u}\\
	\v_{n,k}^{i+1} = &
	\begin{cases}
	\v^{i+1}_{\n_i', k},&\hspace{1.mm} \mbox{if $n=\n'_i$} \\
	\v_{n,k}^{i}, & \hspace{1.mm}\mbox{otherwise}
	\end{cases}\label{update.v}
	\end{align}
	\mbox{\bf End}\\[-2mm]
	\noindent\HRule\vspace{-2mm}
\end{table}

\section{Algorithm Analysis}
\subsection{Delayed gradient calculations}
First, we have to point out that the pipeline solution is not equivalent to the two-time-scale solution. It is observed that the gradient used in the stochastic gradient descent step (\ref{w-algorithm4}) at iteration $i$ is $\grad_z Q(\z^{(J)}_{\n_i',k};\gamma_{\n_i'})$ where $\n_i'=\n_{i-J+1}$, which is $J$-time out-of-date. The cause of the delay is that the variable $\z_{n,k}$ has to conduct $J$ updates in the pipeline. When $\z_{\n_i,k}^{(J)}$ pops up from the pipeline, the iteration index has arrived to $i+J-1$. As a result, the pipeline solution introduces delays in the gradient. 
{\color{black}
	Due to this delay, it does not necessarily follow that deeper pipelines lead to better performance. Actually there is a trade off between depth and performance. 

Fortunately, problems involving delays in gradient calculations are  well-studied in the distributed machine learning and signal processing literature \cite{agarwal2011distributed,peng2016arock,wu2016decentralized}. These works show that convergence can still occur albeit the stability ranges are affected.}
Although, in most literature, the delayed gradient is usually caused by unbalanced computation loads or fragile communication environments instead of pipeline structure, the proof in this paper is inspired from these investigations to some extent. 

\subsection{Convergence Analysis}
To establish the convergence theorem, we require two assumptions
\begin{assumption}[\sc Risk Function]\label{assumption-cost}
	The loss function $Q(z;\gamma_n)\define Q(h_n\tran w; \gamma_n)$ is differentiable, and has an $L$-Lipschitz continuous gradient with respect to $w$ and a $\delta$-Lipschitz continuous gradient with respect to $z$, for every $n=1,\ldots,N$, i.e., for any $w_1, w_2 \in \real^M$:
	\eq{\label{eq-ass-cost-lc-e}
		&\hspace{-3mm}\|\grad_w Q(h_n\tran w_1;\gamma_n) \!-\! \grad_w Q(h_n\tran w_2;\gamma_n)\|\nn\\
		&=\|\grad_z Q(h_n\tran w_1;\gamma_n)h_n \!-\! \grad_z Q(h_n\tran w_2;\gamma_n)h_n\|\nn\\
		&\le L \|w_1-w_2\|\\
		&\hspace{-3mm}\|\grad_z Q(z_1;\gamma_n) - \grad_z Q(z_2;\gamma_n) \| \le \delta \|z_1-z_2\|
	}
	where $L > 0 {\rm\ and\ }\delta>0$. {\color{black}For the regularization term $r(w)$, it is convex and has $\eta$-Lipschitz continuous gradient:
	\eq{
		&\hspace{-3mm}\|\grad_w r(w_1) - \grad_w r(w_2) \| \le \eta \|w_1-w_2\|
	}}\noindent
	We also assume that the risk $R(w)$ is $\nu$-strongly convex, namely,
	\eq{\label{eq-ass-cost-sc-e}
		\hspace{-1mm}\Big(\grad_w R(w_1)- \grad_w R(w_2)\Big)\tran (w_1 - w_2) &\ge \nu \|w_1-w_2\|^2
	}\vspace{-0mm}$\hfill\Box$
\end{assumption}
\begin{assumption}[\sc Topology]\label{assumption-topology} The underlying topology is strongly connected, and the combination/{\color{black}weighted adjacency} matrix 
$A$ is symmetric and doubly stochastic, i.e.,\vspace{-0.5mm}
\eq{
	A = A\tran\;\mbox{and}\;\; A \mathds{1}_K = \mathds{1}_K
}
where $\mathds{1}$ is a vector with all unit entries. {\color{black} We further assume that $a_{kk}>0$ for at least one agent $k$.} 
 $\hfill\Box$
\end{assumption}
Under assumption \ref{assumption-topology}, we can show that matrix $A$ is primitive\cite{horn1990matrix, sayed2014adaptation} and that the second largest magnitude of the eigenvalue of $A$, denoted by $\lambda$, satisfies\cite{pillai2005perron}:
\eq{
	0<\lambda<1
}
\begin{theorem}[\sc\! Convergence of PVRD$^2$]\!\label{theorem.1}\!
	Algorithm \!\! PVRD$^2$\! converges at a linear rate for sufficiently small step-sizes $\mu$, i.e.,
	\eq{
		\Ex \|\w_{i,k} - w^\star\|^2 \leq \rho^i C,\;\;\;\forall k, i>0
	}
	for some constant $C$, where:
	\eq{
		\rho = \max\left(1 - \frac{1-\lambda^J}{2N}, 1-\mu\nu/5\right) \label{converge.rate}
	}
\end{theorem}
 \begin{proof} See Appendix \ref{app.a}. \end{proof} 
 
This theorem indicates that the convergence rate of the algorithm depends on the network topology through $\lambda$, the depth of pipeline $J$, and the strong convexity parameter $\nu$. When $J$ is not very large and the first term in \eqref{converge.rate} is larger than the second one, the convergence rate will depend more on the network topology and the depth of the pipeline. However, when $J$ is large enough so that the second term is dominant, the algorithm performance will depend  on the strong convexity parameter $\nu$, as in the single agent case. 
{\color{black} Compared with other distributed learning algorithms\cite{yuan2017exact2,shi2015extra,kar2011convergence}, the convergence behavior of PVRD$^2$ closer to the behavior of single agent SGD than multi-agent SGD; this is also clear from recursion \eqref{3g.hewsd} in the proof. } 

{\color{black}{\bf{Remark}}: The theorem indicates that the algorithm converges to the minimizer of empirical risk \eqref{eq.cost-of-sum} exactly {\color{black} with linear rate}. In many cases we are instead interested in the minimizer of the average risk $R_E(w) = \Ex_\x Q(w;\x)+r(w)$. By an ergodicity argument, the two minimizers will become closer when the number of data samples increases. \hfill\ensuremath{\square}
}

Notice that  algorithm VRD$^2$ is a special case of algorithm  PVRD$^2$ by setting $J=1$. Thus, we establish the following corollary from Theorem \ref{theorem.1}.
\vspace{1mm}

\begin{corollary}[\sc Convergence of VRD$^2$]\!\label{corollary.1}
	Algorithm\!  VRD$^2$ converges at a linear rate for sufficiently small step-sizes $\mu$, i.e.,
	\eq{
		\Ex \|\w_{i,k} - w^\star\|^2 \leq \rho^i C,\;\;\;\forall k, i>0
	}
	for some constant $C$, where:
	\eq{
		\rho = \max\left(1 - \frac{1-\lambda}{2N}, 1-\mu\nu/4\right) \label{converge.rate.2}
	}
\end{corollary}
\qd

\section{Simulations}\vspace{-0mm}
We illustrate the performance of the PVRD$^2$ algorithm on the MNIST dataset, which consists of 50000 $28\times 28$ handwritten digits\footnote{\url{http://yann.lecun.com/exdb/mnist/}}. In the simulation, we consider the classification task of predicting digit $0$ or digit $1$. We separate the features over 8 networked agents. The loss function we use is logistic regression:
\eq{
	R(w) \!=\! \frac{1}{N} \sum_{n=1}^N \ln\left(\!1+\exp\Big(\!\!-\!\gamma_n \!\sum_{k=1}^Kh_{n,k}\tran w_k\Big) \!\right)
	+ \rho\sum_{k=1}^K\|w_k\|^2
}
In the simulation, we set $\rho = 1\times 10^{-4}$.
From two subplots in Fig.~\ref{fig.visual.w},\! we see that each agent is in charge of part of $w$, and each converges to its corresponding part of~$w^\star$. Next, we compare our algorithm to the method proposed in \cite{sundhar2012new} with some modification, which can be viewed as the deterministic full-gradient version of our algorithm without pipeline. {\color{black}To make a fair comparison, we plot the convergence curve based on the count of gradients calculated and the combination step in Fig.~\ref{fig.converge}. Notice that when the pipeline step $J$ is larger, we need to do more operation on the combination step (\ref{dist.comb}). Therefore, we use a different mini-batch $B$ for different $J$ but we keep the sum, $J+B$  equal 30.}
The curve shows that the larger $J$ we set, the faster the algorithm converges until it is large enough to trigger the second term in the convergence rate.\vspace{-0mm}
\begin{figure}[h]
	\centering
	\includegraphics[scale=0.24]{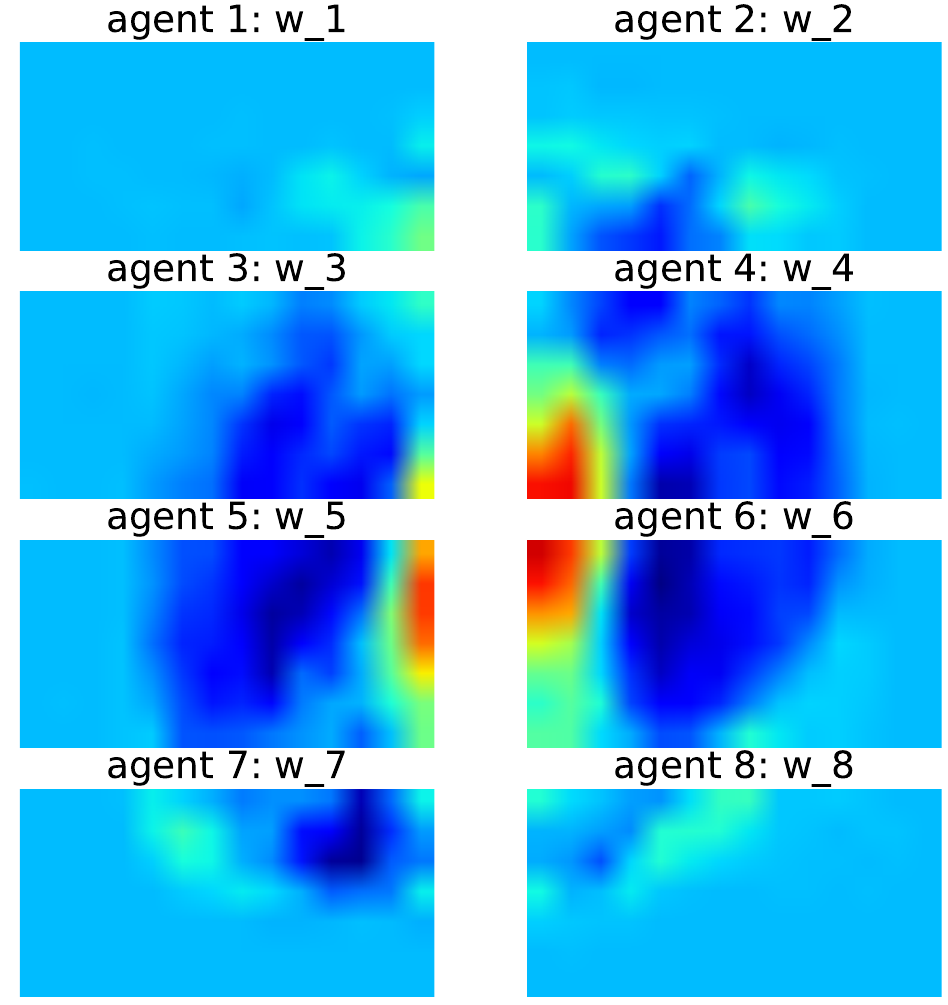}\hspace{0mm}
	\includegraphics[scale=0.25]{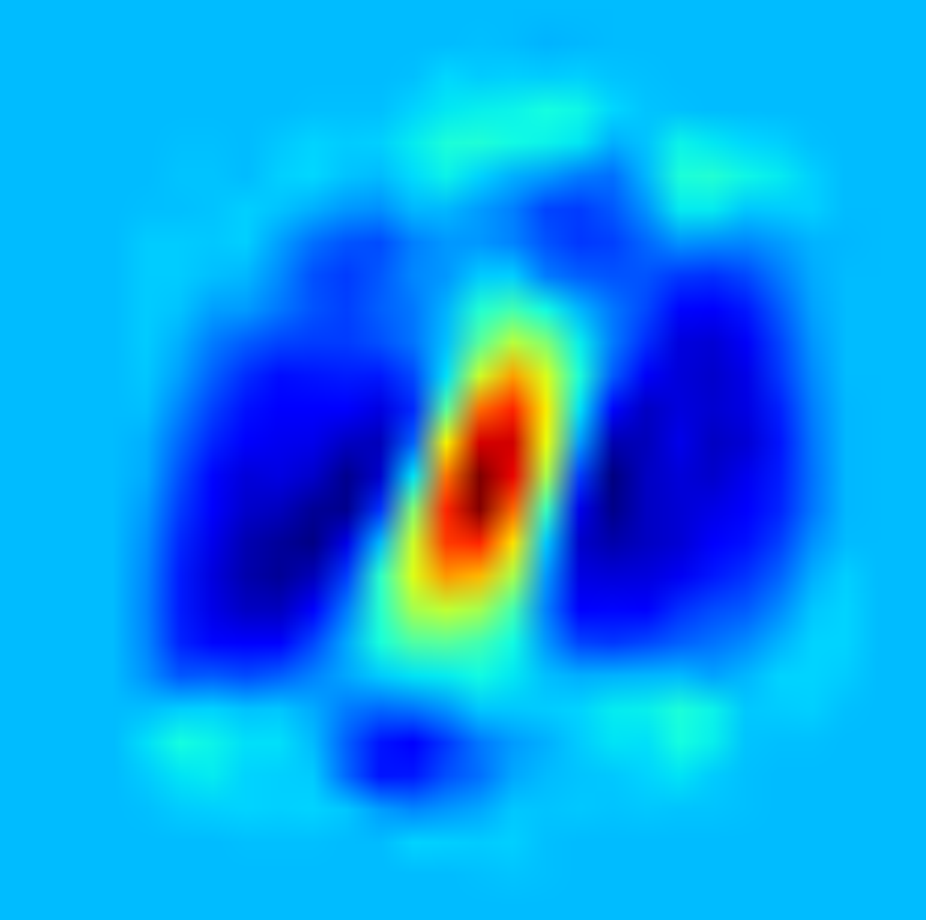}\vspace{0.6mm}
	\caption{Left: Visualization of $\w_{i,k}$;\; Right: Visualization of $w^\star$.\;  }\label{fig.visual.w}\vspace{-2mm}
\end{figure}
\begin{figure}[htp]
	\centering
	\includegraphics[scale=0.55]{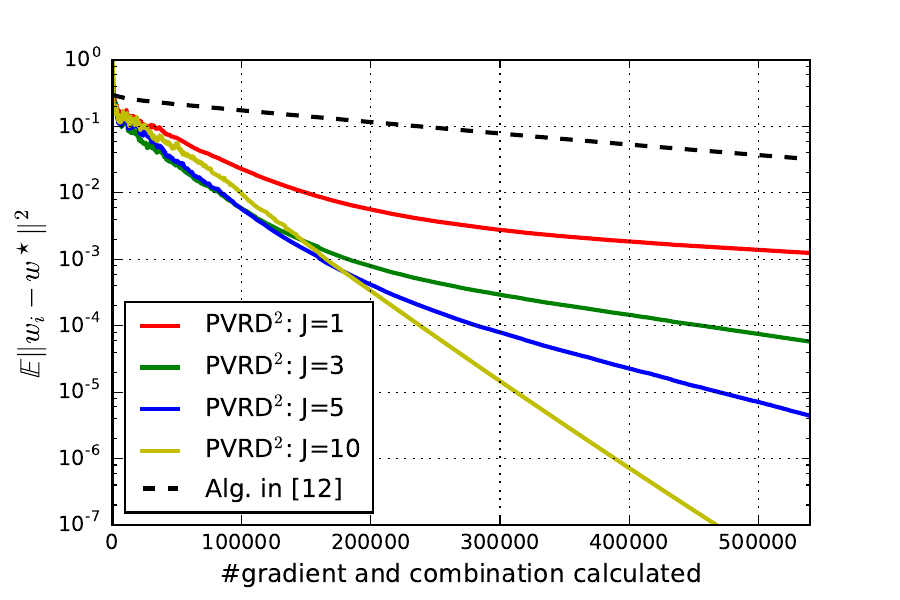}\vspace{-0mm}
	\caption{\color{black} Illustration of the convergence behavior with different pipeline depth $J$ over MNIST dataset}\label{fig.converge}\vspace{-3mm}
\end{figure}
\begin{figure}[htp]
	\centering
	\includegraphics[scale=0.5]{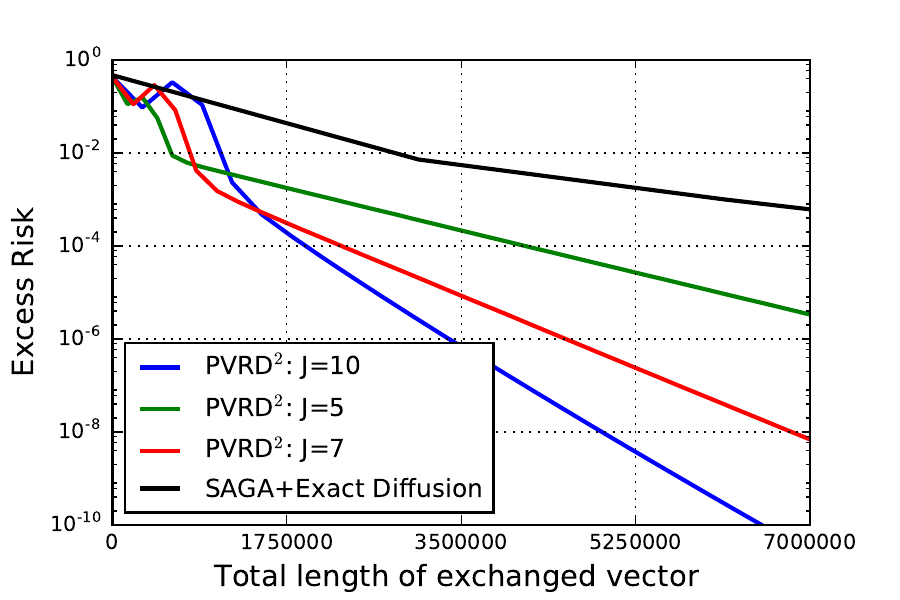}\vspace{-0mm}
	\caption{Comparison of communication vector length between PVRD$^2$ and SAGA+Exact Diffusion over CIFAR-10 Dataset.}\label{fig.communication}\vspace{-5mm}
\end{figure}

Next, we compare the communication cost  on the CIFAR-10 dataset, which consists of 50000 32x32 color images in 10 classes, with 5000 images per class\footnote{\url{https://www.cs.toronto.edu/~kriz/cifar.html}}. For this multi class problem, we use the softmax as the loss function:
\eq{
	R(W) =& -\frac{1}{N}\sum_{n=1}^N\ln\left( \frac{\exp(\sum_{k=1}^KW_{k}[\,:\,,\gamma_n]\tran h_{n,k}) }{\sum_{c=1}^C\exp(\sum_{k=1}^K W_{k}[\,:\,,c]\tran h_{n,k}) }\right) \nn\\
	&\;\;+ \rho \sum_{k=1}^K\|W_{k}\|_F^2
}
where $\|\cdot\|_F$ represents the Frobenius norm, $W_k\in\real^{M\times C}$ is a matrix that has dimensions $M\times C$, where $C=10$ in CIFAR-10 problem, and  $W[\,:\,,c] \mbox{ is the $C$-th column of matrix }W$. Notice that this model is a little bit different from \eqref{cost.full}. However, it is still easy to adapt our algorithm for the softmax cost. Notice that  we now need to find the summation of $C$-classes, i.e.
\eq{
	z_{n,k} \define \ba{c}
	z_{n,k}(1)\\[1mm]
	z_{n,k}(2)\\
	\vdots\\[1mm]
	z_{n,k}(C)
	\ea
	\define \ba{c}
	 \sum_{k=1}^KW_{k}[\,:\,,1]\tran h_{n,k}\\[1.5mm]
	 \sum_{k=1}^KW_{k}[\,:\,,2]\tran h_{n,k}\\
	\vdots\\[1.5mm]
	 \sum_{k=1}^KW_{k}[\,:\,,c]\tran h_{n,k}
	\ea
}
As long as the agent has the information of $z_{n,k}$, it can compute the gradient locally.  For simplicity, we let 
\eq{
	{\rm prob}(z_{n,k}, \gamma_n)\define&\frac{\exp\Big(z_{n,k}(\gamma_n)\Big) }{\sum_{c=1}^C\exp\Big(z_{n,k}(c)\Big) } \nn\\
	=&\frac{\exp\Big(\sum_{k=1}^KW_{k}[\,:\,,\gamma_n]\tran h_{n,k}\Big) }{\sum_{c=1}^C\exp\Big(\sum_{k=1}^K W_{k}[\,:\,,c]\tran h_{n,k}\Big) } 
}
After some algebraic manipulations, we know that
\eq{
	\nabla_{W_k[\,:\,,\gamma_n]} R_n(W) \!=& \Big({\rm prob}(z_{n,k}, \gamma_n)-1\Big)h_{n,k} + \rho W_k[\,:\,,\gamma_n]
}
and for the other columns where  $c\neq\gamma_n$, we have
\eq{
	\nabla_{W_k[\,:\,,c]} R_n(W) =&\; {\rm prob}(z_{n,k}, c)h_{n,k} + \rho W_k[\,:\,,c]
}
Thus, the only modification that is necessary to our algorithm is changing from sending a scalar $z_{n,k}$ into sending a vector $z_{n,k}$. In the simulation, we compare this implementation it with SAGA+Exact diffusion as proposed in \cite{yuan2017efficient}. In the simulation, as plotted in Fig.~\ref{fig.communication}, we set $K=10$, $J=10$, $\rho=1e-4$ and mini-batch is $B=10$. In Fig.~\ref{fig.communication}, the curve is plotted based on the total number of communicated length versus the excess risk, i.e., $R(w_i) - R(w^\star)$. More specifically, for PVRD${}^2$ algorithm, at every iteration, each edge of the network will communicate a length of $J\times C\times B = 1000$. Meanwhile, for Exact Diffusion, it needs $32\times32\times3\times10 = 30720$. Hence, from  Fig.~\ref{fig.communication}, it is not surprising to find that  PVRD${}^2$ is more communication-efficient.

{\color{black}
Lastly, we provide a simulation to show the influence of network structure on  algorithm performance in Fig.~\ref{fig.converge_network}. The approach we utilized to generate the network is the random geometric graph model, which places $N$ nodes uniformly at random in the unit square and two nodes are joined by an edge if the Euclidean distance between the nodes is less than some radius threshold\cite{penrose2003random}. The simulation problem is the same as the previous MNIST problem except we distribute the feature over 28 agents. We generate four network topologies, whose Euclidean distance threshold $\Delta$ are $0.3$, $0.4$, $0.6$, and $\sqrt{2}$---full connected networks, respectively. For all network topologies, we fix the pipeline depth at 20 and  mini-batch at 10. Figure~\ref{fig.converge_network} confirms the conclusion from Theorem \ref{theorem.1} that the denser the topology is, the faster it converges. 

\begin{figure}[htp]
	\centering
	\includegraphics[scale=0.55]{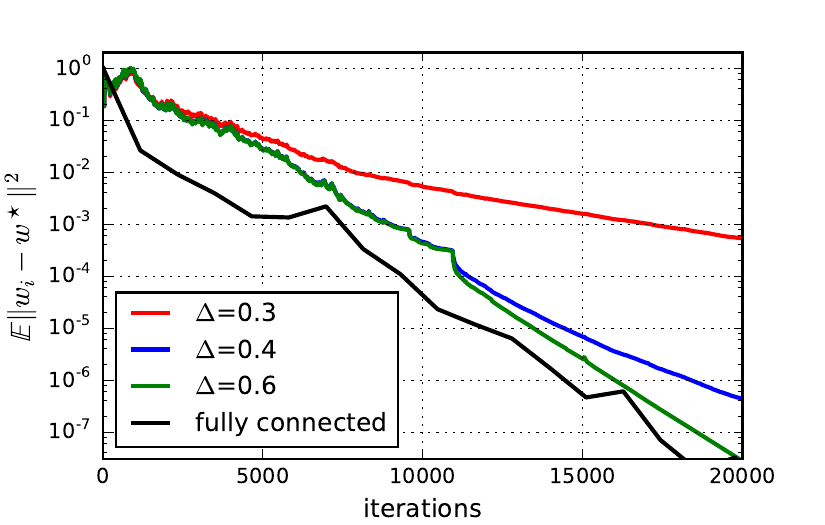}\vspace{-0mm}
	\caption{{\color{black} Comparison of PVRD$^2$ with different network topologies over MNIST Dataset.}}\label{fig.converge_network}\vspace{-5mm}
\end{figure}

}

\appendices
\section{Proof of Theorem \ref{theorem.1}} \label{app.a}
{\color{black}
\subsection{Supporting lemma and proof sketch}
Before we start the proof of the main theorem, we first give two useful lemmas here. 
\begin{lemma} \label{lemma.1}
	For a $\nu$-strongly convex and $L$-Lipschitz gradient continuous function $f(\cdot)$, the following inequality holds:
	\eq{
		(\nabla f(z)-\nabla f(y))\tran(x-z) \leq\;&  \frac{L}{2}\|x-y\|^2- \frac{\nu}{2}\|z-y\|^2\nn\\
		&\;\;{}- \frac{\nu}{2}\|x-z\|^2 \label{cross.grad}
	}
\end{lemma}
{\bf Proof}:
From  the convexity and Lipschitz conditions, we have:
\eq{
	f(x) \leq& f(y) + \nabla f(y)\tran(x-y) + \frac{L}{2}\|x-y\|^2\\
	f(z) \geq& f(y) + \nabla f(y)\tran(z-y) + \frac{\nu}{2}\|z-y\|^2\\
	f(x) \geq& f(z) + \nabla f(z)\tran(x-z) + \frac{\nu}{2}\|x-z\|^2
}
Combining these inequalities, we get:
\eq{
	f(x) - f(z)\leq& f(y) + \nabla f(y)\tran(x-y) + \frac{L}{2}\|x-y\|^2\nn\\
	&\hspace{5mm} - f(y) - \nabla f(y)\tran(z-y) - \frac{\nu}{2}\|z-y\|^2\nn\\
	\leq& \nabla f(y)\tran(x-z) + \frac{L}{2}\|x-y\|^2- \frac{\nu}{2}\|z-y\|^2\\
	f(x) - f(z)\geq &\nabla f(z)\tran(x-z) + \frac{\nu}{2}\|x-z\|^2  
}
Combining the above  two inequalities and rearranging terms, we establish \eqref{cross.grad}.\hfill\qd

\begin{lemma} \label{lemma.exp}
	Consider the inequality recursion for non-negative numbers $\alpha_i$:
	\eq{
		\alpha_{i+1} \leq \lambda \alpha_{i} + \rho^i C \label{23gb}
	}
	where $1>\rho>\lambda > 0$. The sequence $\alpha_i$ converges to zero according to
	\eq{
		\alpha_{i+1} \leq \rho^{i+1}\left( \alpha_0 + \frac{1}{\rho-\lambda} C\right)
	}
\end{lemma} 
{\bf Proof}: Iterating, we get
\eq{
	\alpha_{i+1} \leq& \lambda^2 \alpha_{i-1} + \lambda \rho^{i-1}C + \rho^i C\nn\\
	\leq&\lambda^3 \alpha_{i-2} +\lambda^2 \rho^{i-2}C +  \lambda \rho^{i-1}C + \rho^i C\nn\\
	&\vdots\nn\\[-2mm]
	\leq& \lambda^{i+1} a_0 + \sum_{j=0}^{i}\lambda^j \rho^{i-d} C\nn\\
	\leq& \lambda^{i+1} \alpha_0 + \frac{\rho^i}{1-\lambda/\rho} C\nn\\
	\leq& \rho^{i+1} \left( \alpha_0 + \frac{1}{\rho-\lambda} C\right)
}
\hfill\qd

The procedure of proving the linear convergence in Theorem \ref{theorem.1} consists of two steps. The first step is to construct an error recursion for $\Ex\|\w_{i}-w^\star\|^2$. However, because of many auxiliary variables, eventually, we establish six recursions--- \eqref{recursion1}, \eqref{recursion2}, \eqref{recursion3}, \eqref{recursion4}, \eqref{recursion5}, and \eqref{recursion6}. Note that a simple energy or Lyapunov function cannot establish the linear convergence easily due to the delay introduced by the pipeline strategy. Instead, we employ a mathematical induction method. After assuming the quantity $\Ex\|\w_{i} - \w_{i-1}\|^2$ converges linearly, the six error recursions will have the same format as \eqref{23gb} in Lemma \ref{lemma.exp}. So that we conclude the proof by showing $\Ex\|\w_{i} - \w_{i-1}\|^2$ converges linearly in view of Lemma \ref{lemma.exp}.
}
\subsection{Error Recursion}
In order to shorten the notation, we let $d = J - 1$, such that $\n' = \n_{i-d}$.
First, we rewrite the main recursion in vector notation:
\eq{
	\w_{i+1} 	=\;& \w_{i} - \mu \Big(\big[\nabla_z Q({\sz}^{(J)}_{\n', i-d}; \gamma_{\n'})- \nabla_z Q(\su^i_{\n'}; \gamma_{\n'})\big]h_{\n'} \nn\\
	&\hspace{15mm}+\frac{1}{N}\sum_{n=1}^N \nabla_z Q(\su^i_{n}; \gamma_{n})h_n + \nabla_w r(\w_i)\Big) \label{3g.hewsd}
}
where:
\eq{
	{\sz}^{(J)}_{\n',i-d}\!\!\define\!& {\rm col}\{\z^{(J)}_{\n',1}, \z^{(J)}_{\n',2}, \cdots, \z^{(J)}_{\n',K} \} {\color{black} \in \real^{K\times 1}}\\
	{\su}^i_{\n'}\!\!\define\!& {\rm col}\{\u^{i}_{\n',1}, \u^{i}_{\n',2}, \cdots, \u^{i}_{\n',K} \} {\color{black} \in \real^{K \times 1}}\\
	\hspace{-3mm}\nabla_z Q({\sz}^{(J)}_{\n',i-d}; \gamma_{\n'})\!\!\define\!&{\rm blockdiag}\{\nabla_z Q(\z^{(J)}_{\n', 1}; \gamma_{\n'})I_{M_1},\cdots,\nn\\
	&\hspace{16mm} \nabla_z Q(\z^{(J)}_{\n', K}; \gamma_{\n'})I_{M_K}\} {\color{black} \in \real^{M\times M}} \\
	\hspace{-3mm}\nabla_z Q(\su^{i}_{\n'}; \gamma_{\n'})\!\!\define\!&{\rm blockdiag}\{\nabla_z Q(\u^{i}_{\n', 1}; \gamma_{\n'})I_{M_1},\cdots,\nn\\
	&\hspace{16mm} \nabla_z Q(\u^{i}_{\n',K}; \gamma_{\n'})I_{M_K}\} {\color{black} \in \real^{M\times M}}
}
In our notation convention, all calligraphic symbols represent stacked vectors or matrices over agents.
Notice the subscript $i-d$ in $\sz$ is used to represents the $\z$ is computed from $\w_{i-d}$ instead of data or agent index. 
Introducing the error quantity $\tw_i \define w^\star - \w_i$, we obtain the error recursion:
\eq{
	\tw_{i+1} 
	=\;& \tw_{i} + \mu \Big(\big[\nabla_z Q({\sz}_{\n', i-d}^{(J)}; \gamma_{\n'})- \nabla_z Q(\su^i_{\n'}; \gamma_{\n'})\big]h_{\n'}\nn \\
	&\hspace{15mm} {}+ \frac{1}{N}\sum_{n=1}^N \nabla_z Q(\su^{i}_{n}; \gamma_{n})h_{n} + \nabla_w r(\w_i^k)\Big)
}
We introduce  the filtration $\cF_{i}$, which consists of all previous information $\w_{j}, \;j\leq i$ and indices $\n_{j},\;j\leq i-J$. Then, the modified gradient step satisfies the unbiasedness property:
\eq{
	&\Ex_{\n'}\Bigg[\Big(\big[\nabla_z Q({\sz}_{\n',i-d}^{(J)}; \gamma_{\n'})- \nabla_z Q(\su^i_{\n'}; \gamma_{\n'})\big]h_{\n'} \nn\\
	&\hspace{1cm}+\frac{1}{N}\sum_{n=1}^N \nabla_z Q(\su^{i}_{n}; \gamma_{n})h_{n} + \nabla_w r(\w_i)\Big)\Big|\cF_{i}\Bigg]\nn\\
	&= \frac{1}{N}\sum_{n=1}^N \nabla_z Q({\sz}^{(J)}_{n,i-d}; \gamma_{n}) h_n+ \nabla_w r(\w_i)
}
Thus, we have the following error recursion:
\eq{
	&\hspace{-4mm}\Ex_{\n'} [\|\tw_{i+1}\|^2 |\cF_{i}]\nn\\
	=& \|\tw_{i}\|^2 + \frac{2\mu}{N}\sum_{n=1}^N \tw_{i}\tran [\nabla_z Q({\sz}^{(J)}_{n,i-d};\gamma_n) h_n + \nabla_w r(\w_i)] \nn\\
	&\hspace{3mm}{}+ \mu^2\Ex\Bigg[\Big\|\big[\nabla_z Q(\sz^{(J)}_{\n', i-d}; \gamma_{\n'})- \nabla_z Q(\su^{i}_{\n'}; \gamma_{\n'})\big]h_{\n'} \nn\\
	&\hspace{17mm}+ \frac{1}{N}\sum_{n=1}^N \nabla_z Q(\su^{i}_{n}; \gamma_{n})h_{n}  + \nabla_w r(\w_i)\big\|^2 \Big|\filt_i\Bigg]\nn\\
	\leq&\;\|\tw_{i}\|^2 \nn\\
	&\hspace{1mm} {} {\color{black}+ \frac{2\mu}{N}\sum_{n=1}^N \tw_{i}\tran [\nabla_z Q({\sz}^{(J)}_{n,i-d};\gamma_n)h_n  + \nabla_w r(\w_i) ]}\nn\\
	&\hspace{1mm}{} + 4\mu^2\Ex\Big\|\big[\nabla_z Q({\sz}^{(J)}_{\n',i-d}; \gamma_{\n'})\!-\! \nabla_z Q(h_{\n'}\tran w^\star; \gamma_{\n'})I_M\big]h_{\n'}\Big\|^2\nn\\
	&\hspace{1mm}{} + 4\mu^2\Ex\Big\|\big[\nabla_z Q({\su}^{i}_{\n'}; \gamma_{\n'})\!-\! \nabla_z Q(h_{\n'}\tran w^\star; \gamma_{\n'})I_M\big]h_{\n'}\Big\|^2\nn\\
	&\hspace{1mm}{} + 4\mu^2\frac{1}{N}\sum_{n=1}^N \Big\|[\nabla_z Q(\su^{i}_{n}; \gamma_{n})\!-\!\nabla_z Q(h_n\tran w^\star; \gamma_{n})I_M]h_{n}  \Big\|^2\nn\\
	&\hspace{1mm}{} + 4\mu^2\Big\|\nabla_w r(\w_i)-\nabla_w r(w^\star)  \Big\|^2
	\label{eq.recursion1}
}
where we appealed to Jensen's inequality and $I_M$ is the $M\times M$ identity matrix. 
Next, we focus on the cross term:
{\color{black}
\eq{
	&\hspace{-4mm}\frac{1}{N}\sum_{n=1}^N \tw_{i}\tran [\nabla_z Q({\sz}^{(J)}_{n,i-d};\gamma_n)h_n + \nabla_w r(\w_i)] \nn\\
	=&\frac{1}{N}\sum_{n=1}^N \tw_{i}\tran [\nabla_z Q({\sz}^{(J)}_{n,i-d};\gamma_n)h_n + \nabla_w r(\w_i) - \nabla_w R(w^\star)] \nn\\
	\stackrel{(a)}{=}&\tw_i\tran \left[\frac{1}{N}\sum_{n=1}^N\nabla_z Q(\bar{\sz}_{n,i-d};\gamma_n) h_n + \nabla_w r(\w_i) \!-\!  \nabla_w R(w^\star)\right]
	\nn\\
	&+\frac{1}{N}\sum_{n=1}^N\tw_{i}\tran[\nabla_z Q({\sz}^{(J)}_{n,i-d};\gamma_n) h_n \!-\!  \nabla_z Q(\bar{\sz}_{n,i-d};\gamma_n) h_n]\nn\\
	\stackrel{(b)}{=} & [\nabla_w R(\w_{i-d}) - \nabla_w R(w^\star)]\tran(w^\star -  \w_{i})
	\nn\\
	& + [\nabla_w r(\w_{i}) - \nabla_w r(\w_{i-d})]\tran(w^\star -  \w_{i})\nn\\
	&+\frac{1}{N}\sum_{n=1}^N\tw_{i}\tran[\nabla_z Q({\sz}^{(J)}_{n,i-d};\gamma_n) h_n \!-\!  \nabla_z Q(\bar{\sz}_{n,i-d};\gamma_n) h_n] \label{cross.term}
}}
where we define
\eq{
	\bar{\z}_{n,i-d} \define & \frac{1}{K}\mathds{1}_K\tran {\sz}^{(0)}_{n,i-d} {\color{black} \in \real} \\
	\bar{\sz}_{n,i-d}\define& {\rm col}\{\bar{\z}_{n,i-d},\cdots , \bar{\z}_{n,i-d}\}=\bar{\z}_{n,i-d}\mathds{1}_K  {\color{black} \in  \real^{K\times1}}
}
In step (a), we add and subtract the same term about $\nabla_z Q(\bar{\sz}_{n,i-d};\gamma_n)h_n$; In step (b), we exploit the property that
$\bar{\z}_{n,i-d} = h\tran_{n}\w_{i-d}$. Indeed, note that 
\eq{
	\bar{\z}_{n,i-d}  =& \frac{1}{K}\mathds{1}_K\tran {\sz}^{(0)}_{n,i-d} \nn\\
	\stackrel{\eqref{pipeline.input}}{=}& \frac{1}{K}\sum_{k=1}^K(Kh_{n,k}\tran \w_{i-d,k} + \u^{i-d}_{n,k} - \v^{i-d}_{n,k} )\nn\\
	\stackrel{\eqref{dd.unbiased}}{=}& \sum_{k=1}^Kh_{n,k}\tran \w_{i-d,k}\nn\\
	=& h_{n}\tran \w_{i-d} \label{equiv.z.bar}
}
{\color{black}
Now, to bound the first term in \eqref{cross.term}, we apply the result (\ref{cross.grad}) from Lemma \ref{lemma.1} so that we immediately get}
\eq{
	&\hspace{-5mm}[\nabla_w R(\w_{i-d}) - \nabla_w R(w^\star)]\tran(w^\star - \w_{i}) \nn\\
	\leq&  \frac{L+\eta}{2}\|\w_{i}-\w_{i-d}\|^2- \frac{\nu}{2}\|w^\star-\w_{i-d}\|^2- \frac{\nu}{2}\|\w_{i}-w^\star\|^2\nn\\
	\leq&  \frac{(L+\eta)d}{2}\sum_{j=i-d}^{i-1}\|\w_{j+1}-\w_{j}\|^2- \frac{\nu}{2}\|\tw_{i-d}\|^2- \frac{\nu}{2}\|\tw_{i}\|^2 \label{eq.result4}
}
{\color{black} and we have
\eq{
	 &[\nabla_w r(\w_{i}) - \nabla_w r(\w_{i-d})]\tran(w^\star -  \w_{i}) \nn\\
	 &\leq  \frac{\eta\epsilon}{2}\|\w_{i} - \w_{i-d}\|^2 + \frac{1}{2\epsilon} \|\tw_{i}\|^2\nn\\
	 &\leq  \frac{\eta d}{\nu}\sum_{j=i-d}^{i-1}\|\w_{j+1}-\w_{j}\|^2 + \frac{\nu}{4} \|\tw_{i}\|^2
}}\noindent
where the first inequality is due to Young's inequality $a\tran b \leq \frac{\epsilon}{2}\|a\|^2 + \frac{1}{2\epsilon}\|b\|^2$ and $\epsilon$ can be any positive number. And in second step we set $\epsilon=\frac{2}{\nu}$.

\noindent For the third term in \eqref{cross.term}, we obtain:
\eq{
	&\hspace{-4mm}\tw_{i}\tran[\nabla_z Q({\sz}^{(J)}_{n,i-d};\gamma_n) h_n - \nabla_z Q(\bar{\sz}_{n,i-d};\gamma_n) h_n]\nn\\
	\leq& \frac{\epsilon}{2}\|\tw_{i}\|^2+\frac{1}{2\epsilon}
	\Big\|\nabla Q({\sz}^{(J)}_{n,i-d};\gamma_n) h_n - \nabla Q(\bar{\sz}_{n,i-d};\gamma_n) h_n\Big\|^2\nn\\ 
	\leq& \frac{\epsilon}{2}\|\tw_i \|^2+\frac{\delta^2}{2\epsilon}
	\Big\|{\sz}^{(J)}_{n,i-d}  - \bar{\sz}_{n,i-d}\|^2\|h_n\|^2 \label{eq.result5}
}

Next, for the remaining terms in the main recursion \eqref{eq.recursion1}, we have the following bounds through the Lipschitz condition:
\eq{
	&\hspace{-5mm}\Ex_{\n'}\Big\|\big[\nabla_z Q(\sz^{(J)}_{\n',i-d}; \gamma_{\n'})- \nabla_z Q(h_{\n'}\tran w^\star; \gamma_{\n'}I_M)\big]h_{\n'}\Big\|^2\nn\\
	\leq\;&2\Ex_{\n'}\Big\|\big[\nabla_z Q(\sz^{(J)}_{\n',i-d}; \gamma_{\n'})- \nabla_z Q(\bar{\sz}_{\n',i-d}; \gamma_{\n'})\big]h_{\n'}\Big\|^2\nn\\
	&\;\;+2\Ex_{\n'}\Big\|\big[\nabla_z Q(h_{\n'}\tran \w_{i-d}; \gamma_{\n'})- \nabla_z Q(h_{\n'}\tran w^\star; \gamma_{\n'})\big]h_{\n'}\Big\|^2\nn\\
	\leq\;& \frac{2\delta^2}{N}\sum_{n=1}^{N}\|\sz^{(J)}_{n,i-d} -\bar{\sz}_{n,i-d}\|^2 \|h_n\|^2 +2L^2\|\tw_{i-d}\|^2 \label{eq.result1}
}
and
\eq{
	&\hspace{-7mm}\Ex_{\n'}\Big\|\big[\nabla_z Q({\su}^{i}_{\n'}; \gamma_{\n'})- \nabla_z Q(h_{\n'}\tran w^\star; \gamma_{\n'})I_M\big]h_{\n'}\Big\|^2\nn\\
	\leq\;&
	\Ex \|\nabla_z Q({\su}^{i}_{\n'}; \gamma_{\n'})- \nabla_z Q(h_{\n'}\tran w^\star; \gamma_{\n'})I_M\|^2\|h_{\n'}\|^2\nn\\
	\leq\;&\frac{\delta^2}{N}\sum_{n=1}^N\|{\su}^{i}_{n} - h_{n}\tran w^\star \mathds{1}_K\|^2\|h_n\|^2 \label{eq.result2}
}
and
\eq{
	&\hspace{-7mm}\frac{1}{N}\sum_{n=1}^N \Big\|[\nabla_z Q(\su^{i}_{n}; \gamma_{n})-\nabla_z Q(h_n\tran w^\star; \gamma_{n})I_M]h_{n}  \Big\|^2\nn\\
	\leq &\frac{\delta^2}{N}\sum_{n=1}^N\|{\su}^{i}_{n} - h_n\tran w^\star \mathds{1}_K\|^2\|h_n\|^2. \label{eq.result3}
}
and
\eq{
	4\mu^2 \Big\|[\nabla_w r(\w_i)-\nabla_w r(w^\star)  \Big\|^2\leq 4\mu^2\eta^2\|\tw_i\|^2\label{eq.result6}
}
Substituting (\ref{eq.result4}) -- (\ref{eq.result6}) into the main error recursion (\ref{eq.recursion1}), we have:
\eq{
	&\hspace{-3mm}\Ex_{\n'} [\|\tw_{i+1}\|^2 |\cF_{i}]\nn\\
	\leq & (1+4\mu^2\eta^2)\|\tw_{i}\|^2 +  \mu d(L+\eta+\frac{2\eta}{\nu} )\sum_{j=i-d}^{i-1}\|\w_{j+1}-\w_{j}\|^2\nn\\
	&\;\;{}-\nu\mu\|\tw_{i-d}\|^2- \frac{\mu\nu}{2}\|\tw_{i}\|^2+\mu\epsilon\|\tw_i \|^2\nn\\
	&\;\;{}+\left(\frac{\mu\delta^2}{\epsilon} + 8\mu^2 \delta^2\right)\frac{1}{N}\sum_{n=1}^N
	\|{\sz}^{(J)}_{n,i-d}  - \bar{\sz}_{n,i-d}\|^2\|h_n\|^2\nnb
	&\;\;{}+8\mu^2 L^2\|\tw_{i-d}\|^2+\frac{8\mu^2\delta^2}{N}\sum_{n=1}^N\|{\su}^{i}_{n} - h_n\tran w^\star \mathds{1}_K\|^2\|h_n\|^2	
}
Let $\epsilon = \frac{\nu}{6}$ and denote $L' \define L+\eta+\frac{2\eta}{\nu}$. Rearranging terms we get:
\eq{
	&\hspace{-3mm}\Ex_{\n'} [\|\tw_{i+1}\|^2 |\cF_{i}]\nn\\
	\leq & (1-\mu\nu/3+4\mu^2\eta^2)\|\tw_{i}\|^2+ \mu dL'\sum_{j=i-d}^{i-1}\|\w_{j+1}-\w_{j}\|^2\nn\\
	&\;\;{} - (\nu\mu - 8\mu^2L^2)\|\tw_{i-d}\|^2\!\nn\\
	&\;\;{}+ {\frac{8\mu\delta^2}{\nu N}}\sum_{n=1}^N
	\Big\|{\sz}^{(J)}_{n,i-d} \! -\! \bar{\sz}_{n,i-d}\Big\|^2\|h_n\|^2\nn\\
	&\;\;{}+\frac{8\delta^2\mu^2
	}{N}\sum_{n=1}^N\|{\su}^{i}_{n} - h_n\tran w^\star \mathds{1}_K\|^2\|h_n\|^2\label{recursion1}
}
where we relaxed $ \frac{6\mu\delta^2}{\nu} + 8\mu^2 \delta^2 $ into the upper bound $\frac{8\mu\delta^2}{\nu }$, which requires
\eq{
	\frac{8\mu\delta^2}{\nu} \geq \frac{6\mu\delta^2}{\nu} + 8\mu^2 \delta^2  \Longleftrightarrow \mu\leq \frac{1}{4\nu}
}
So far we have established \eqref{recursion1}. 
Next, we seek a recursion for the  inner difference:
\eq{
	&\hspace{-5mm}\Ex[\|\w_{i+1} - \w_{i}\|^2 | \cF_i] \nn\\
	= &\mu^2\Ex_{\n'}\Bigg[\Big\|\nabla_z Q({\sz}_{\n',i-d}^{(J)}; \gamma_{\n'})- \nabla_z Q(\su^i_{\n'}; \gamma_{\n'})\big]h_{\n'} \nn\\
	&\hspace{1cm}\;\;\;\;+\frac{1}{N}\sum_{n=1}^N \nabla_z Q(\su^{i}_{n}; \gamma_{n})h_{n} + \nabla_w r(\w_i)\Big\|^2\Big|\cF_{i}\Bigg]
}
Applying Jensen's inequality and combining results  \eqref{eq.result1}, \eqref{eq.result2}, \eqref{eq.result3}, and \eqref{eq.result6} we get:
\eq{
	&\hspace{-4mm}\Ex\left[\|\w_{i+1} - \w_{i}\|^2|\cF_{i} \right]\nn\\
	\leq&  \frac{8\mu^2\delta^2}{N}\sum_{n=1}^{N}\|\sz^{(J)}_{n,i-d} -\bar{\sz}_{n,i-d}\|^2 \|h_n\|^2\nn\\
	&\;\;\; + \frac{8\mu^2\delta^2}{N}\sum_{n=1}^N\|{\su}^{i}_{n} - h_n\tran w^\star \mathds{1}_K\|^2\|h_n\|^2 \nn\\
	&\;\;\;+8\mu^2L^2\|\tw_{i-d}\|^2  + 4\mu^2\eta^2\|\tw_i\|^2\label{recursion2}
}
Then, we derive the result:
\eq{
	&\hspace{-5mm}\|{\sz}_{i-d, n}^{(J)} - \bar{\sz}_{i-d, n}\|^2\nn\\
	= &
	\|{\sz}_{i-d, n}^{(J)} - \frac{1}{K}\one\one\tran{\sz}_{i-d, n}^{(J)}\|^2 \nn\\
	\stackrel{\eqref{dist.comb}}{=}&\left\|\left(I-\frac{1}{K}\one\one\tran\right)A^J(KH_n\tran \w_{i-d} + \su^{i}_{n} - \sv^{i}_{n})\right\|^2\nn\\
	\stackrel{(a)}{=}&\left\|\left(A^J-\frac{1}{K}\one\one\tran\right)(KH_n\tran \w_{i-d}- \sv^{i}_{n}+ \su^{i}_{n} -\bar{\su}^{i}_{n})\right\|^2\nn\\
	\stackrel{(b)}{\leq}&\lambda^{J}\| \su^{i}_{n} -\bar{\su}^{i}_{n}\|^2 + \frac{\lambda^{2J}}{1-\lambda^{J}}\left\|KH_n\tran \w_{i-d} - \sv^{i}_{n}\right\|^2
}
{\color{black} where $A^J$ arises from the definition of ${\sz}_{i-d, n}^{(J)} $, which is the stacked vector of $z_{i-d,n}^{(0)}$ after $J$-step consensus, and}
where we denote:
\eq{
	H_n \define \ba{ccc} h_{1,n}&&\\&\ddots&\\&&h_{K,n}\ea {\color{black} \in \real^{M\times K}}
}
Step (a) holds because $\left(A^J-\frac{1}{K}\one\one\tran\right)\bar{\su}^{i}_{n}=0$, and
\eq{
	\bar{\su}^{i}_{n}\define \frac{1}{K}\one\one_K\tran \su_{n}^i = \left(\frac{1}{K}\sum_{k=1}^k \u_{n,k}^i\right)	\one {\color{black} \in \real^{K\times1}}
}In step (b), $\lambda$ is the second largest eigenvalue of $A$.
Multiplying by $\|h_n\|^2$, taking expectation and averaging over $n$:
\eq{
	&\hspace{-4mm}\frac{1}{N}\sum_{n=1}^N \Ex\|{\sz}_{i-d, n}^{(J)} - \bar{\sz}_{i-d, n}\|^2\|h_n\|^2\nn\\
	\leq& \lambda^{J}\frac{1}{N}\sum_{n=1}^N\Ex\| \su^{i}_{n} -\bar{\su}^{i}_{n}\|^2\|h_n\|^2 \nn\\
	&\;\;+ \frac{\lambda^{2J}}{1-\lambda^{J}}\frac{1}{N}\sum_{n=1}^N\Ex\left\|KH_n\tran \w_{i-d} - \sv^{i}_{n}\right\|^2\|h_n\|^2 \label{recursion3}
}
Next, using the uniform sampling property, we establish a similar recursion for  $\su_n^i$:
\eq{
	&\hspace{-7mm}
	\frac{1}{N}\sum_{n=1}^N\Ex_{\n'}\Big[\|{\su}^{i+1}_{n} - \bar{\su}^{i+1}_{n}\|^2\|h_n\|^2\,\big|\,\cF_i\Big]\nn\\
	\stackrel{\eqref{update.u}}{=}& \frac{N-1}{N^2}\sum_{n=1}^N\|\su^{i}_n - \bar{\su}^{i}_{n}\|^2\|h_n\|^2\nn\\
	&\;\;+\frac{1}{N^2}\sum_{n=1}^N\|{\sz}^{(J)}_{n,i-d} - \bar{\sz}_{n,i-d}\|^2\|h_n\|^2\label{recursion4}
}
Similarly, we get:
\eq{&\hspace{-8mm}\Ex_{\n'}\left[\|{\su}^{i}_{n} - h_n\tran w^\star\one\|^2  |\cF_{i} \right] \nn\\
	\stackrel{\eqref{update.u}}{=}\;& \frac{1}{N}\| \sz^{(J)}_{n,i-d}- h_n\tran w^\star\one\|^2+ \frac{N-1}{N}\|{\su}^{i-1}_{n} - h_n\tran w^\star\one\|^2\nn\\
	\leq\;&\frac{2}{N}\| \sz^{(J)}_{n,i-d}- \bar{\sz}_{n,i-d}\|^2+ \frac{2}{N}\|\bar{\sz}_{n,i-d}- h_n\tran w^\star\one\|^2\nn\\
	&\;\; + \frac{N-1}{N}\sum_{n=1}^N\|{\su}^{i-1}_{n} - h_n\tran w^\star\one\|^2\nn\\
	\stackrel{\eqref{equiv.z.bar}}{=}\;&\frac{2}{N}\| \sz^{(J)}_{n,i-d}- \bar{\sz}_{n,i-d}\|^2+ \frac{2}{N}\|h_n\tran \widetilde{\w}_{i-d}\|^2\nn\\
	&\;\; + \frac{N-1}{N}\|{\su}^{i-1}_{n} - h_n\tran w^\star\one\|^2
}
Multiplying $\|h_n\|^2$, taking expectation over filtration and averaging over $n$:
\eq{
	&\hspace{-5mm}\frac{1}{N}\sum_{n=1}^N	\Ex \|{\su}^{i}_{n} - h_n\tran w^\star\one\|^2  \|h_n\|^2\nn\\
	&\leq \frac{2}{N^2}\sum_{n=1}^N\Ex\| \sz^{(J)}_{n,i-d}- \bar{\sz}_{n,i-d}\|^2\|h_n\|^2\nn\\
	&\;\;\;\;\;\;{}+\frac{2}{N^2}\sum_{n=1}^N\Ex\|h_n\|^4 \|\tw_{i-d}\|^2 \nn\\
	&\;\;\;\;\;\;{}\textsl{}+ \frac{N-1}{N^2}\sum_{n=1}^N\Ex\|{\su}^{i-1}_{n} - h_n\tran w^\star\one\|^2\|h_n\|^2\label{recursion5}
}
Lastly, we obtain:
\eq{
	&\hspace{-10mm}\Ex_{\n'}\left[\|KH_{n}\tran \w_{i-d} - \sv^{i}_{n}\|^2 |\cF_{i} \right]\nn\\
	=\;& \frac{1}{N}\|KH_n(\w_{i-d}-\w_{i-d-1})\|^2\nn\\
	&\;\;+ \frac{N-1}{N}\|KH_{n}\tran \w_{i-d}  - \sv^{i-1}_{n}\|^2 \nn\\
	\leq\;& \frac{1}{N}\|KH_n(\w_{i-d}-\w_{i-d-1})\|^2 \nn\\
	&\;\;+ \frac{N-1}{tN}\|KH_n\tran \w_{i-d-1} - \sv^{i-1}_n\|^2 \nn\\
	&\;\;{} + \frac{N-1}{(1-t)N}\|KH_{n}\tran (\w_{i-d} - \w_{i-d-1})\|^2  
}
If we choose $t = \frac{N-1}{N-1/2}$,  then
\eq{
	&\hspace{-6mm}\Ex_{\n'}\left[\|KH_{n}\tran \w_{i-d} - \sv^{i}_{n}\|^2|\cF_{i} \right]\nn\\
	\leq& \frac{1+2(N-1)^2}{N}\|KH_n(\w_{i-d}-\w_{i-d-1})\|^2 \nn\\
	&\;\;{}+ \frac{N-1/2}{N}\|KH_n\tran \w_{i-d-1} - \sv^{i-1}_n\|^2 
}
Notice that $\|H_n\|^2\leq\|H_n\|^2_F = \|h_n\|^2$. After multiplying by $\|h_n\|^2$, taking expectation over the filtration and averaging over $n$, we have
\eq{
	&\hspace{-3mm}\frac{1}{N}\sum_{n=1}^N\Ex\|KH_{n}\tran \w_{i-d} - \sv^{i}_{n}\|^2\|h_n\|^2\nn\\ \leq& \frac{K^2+2K^2(N-1)^2}{N^2}\sum_{n=1}^N\Ex\|(\w_{i-d}-\w_{i-d-1})\|^2\|h_n\|^4\nn\\
	&\;\;{}+ \frac{N-1/2}{N^2}\sum_{n=1}^N\Ex\|KH_n\tran \w_{i-d-1} - \sv^{i-1}_n\|^2\label{recursion6}
}

To simplify the notation, we introduce
\eq{
	a_{i}\define&\Ex\|\tw_{i}\|^2\\
	b_{i}\define& \Ex\|\w_{i} - \w_{i-1}\|^2\\
	c_{i} \define& \frac{1}{N}\sum_{n=1}^N\Ex\!\Big\|{\sz}^{(J)}_{n,i-d}  \!-\! \bar{\sz}_{n,i-d}\|^2\|h_n\|^2\\
	d_{i} \define& \frac{1}{N}\sum_{n=1}^N\Ex\|{\su}^{i}_{n} - \bar{\su}^{i}_{n}\|^2\|h_n\|^2\\
	e_{i} \define& \frac{1}{N}\sum_{n=1}^N\Ex\|{\su}^{i}_{n} - h_n\tran w^\star\one\|^2\|h_n\|^2 \\
	f_{i}\define& \frac{1}{N}\sum_{n=1}^N\Ex\|KH_{n}\tran \w_{i-d} - \sv^{i}_{n}\|^2\|h_n\|^2\\
	h^4 \define& \frac{1}{N}\sum_{n=1}^N \|h_n\|^4
} 
\noindent 
Using the above notation, we have established so far the following six recursions from \eqref{recursion1}, \eqref{recursion2}, \eqref{recursion3}, \eqref{recursion4}, \eqref{recursion5}, \eqref{recursion6}:
\eq{
	a_{i+1} \leq\;& (1-\mu\nu/3 + 4\mu^2\eta^2) a_i + \mu dL'\sum_{j=i-d}^{i-1} b_{j+1}\nn\\
	&\;\; -(\nu\mu - 8\mu^2L^2) a_{i-d}+\frac{8\mu\delta^2}{\nu} c_i + 8\mu^2\delta^2 e_i\label{recursion.a}\\
	b_{i+1} \leq\;& 8\mu^2\delta^2 c_i + 8\mu^2\delta^2 e_i+8\mu^2L^2a_{i-d} + 4\mu^2\eta^2a_i\label{recursion.b}\\
	c_{i} \leq\;& \lambda^J d_i + \frac{\lambda^{2J}}{1-\lambda^{J}}f_i\label{recursion.c}\\
	d_{i+1} =\;&\frac{1}{N} c_i +\frac{N-1}{N}d_{i}\label{recursion.d}\\
	e_{i} = \;&\frac{2}{N}h^4a_{i-d} + \frac{2}{N}c_i + \frac{N-1}{N} e_{i-1}\label{recursion.e}\\
	f_{i} \leq\;& \frac{N-1/2}{N}f_{i-1} + \frac{K^2+2K^2(N-1)^2}{N}h^4b_{i-d} \label{recursion.f}
}

\subsection{Linear Convergence of Error Recursion}
{\color{black}
Finally, we establish convergence of the PVRD$^2$ algorithm using Lemma \ref{lemma.exp} and mathematical induction.}
Assuming $b_i\leq \rho^iC_0$ for $ j \leq i-1$, where we set
\eq{\label{rho}
	\rho = \max\left(1 - \frac{1-\lambda^J}{2N}, 1-\mu\nu/5\right)
}
Then, from \eqref{recursion.f}, we have
\eq{
	f_{i} \leq& \frac{N-1/2}{N}f_{i-1} + \frac{K^2+K^2(N-1)^2}{N\rho^{d-1}}h^4 C_0 \rho^{i-1}\nn\\
	\leq& \rho^{i} \left(f_0 + \frac{1}{\rho - \frac{N-1/2}{N}}\frac{K^2+2K^2(N-1)^2}{N\rho^{d-1}}h^4C_0\right)\nn\\
	\define&\rho^{i} C_1
}
Substituting \eqref{recursion.c} into \eqref{recursion.d}, we have
\eq{
	d_{i+1} \leq& \frac{\lambda^J}{N}d_i + \frac{\lambda^{2J}}{N(1-\lambda^J)}f_i +\frac{N-1}{N}d_i\\
	=&\frac{\lambda^J+N-1}{N}d_i + \frac{\lambda^{2J}}{N(1-\lambda^J)}f_i\nn\\
	\leq&\left(1-\frac{1-\lambda^J}{N}\right)d_i + \frac{\lambda^{2J}}{N(1-\lambda^J)}\rho^i C_1\nn\\
	\leq& \rho^{i+1}\left(d_0 + \frac{1}{\rho - 1 + \frac{1-\lambda^J}{N} }\frac{\lambda^{2J}}{N(1-\lambda^J)}C_1\right)\nn\\
	\define&\rho^{i+1} C_2
}
From \eqref{recursion.c}:
\eq{
	c_i \leq& \lambda^J \rho^i C_2 + \frac{\lambda^{2J}}{1-\lambda^{J}}\rho^{i} C_1\nn\\
	=&\rho^i\left(\lambda^JC_2 + \frac{\lambda^{2J}}{1-\lambda^{J}} C_1\right)\nn\\
	\define&\rho^{i} C_3
}
Adding \eqref{recursion.a} and \eqref{recursion.e} with a positive coefficient $\gamma$, we have
\eq{
	a_{i+1} +\gamma e_{i} \leq\;& (1-\mu\nu/3+4\mu^2\eta^2) a_i + \mu dL'\sum_{j=i-d}^{i-1} b_{j+1}\nn\\
	&\;-(\nu\mu - 8\mu^2L^2) a_{i-d}+\frac{8\mu\delta^2}{\nu} c_i + 8\mu^2\delta^2 e_i\nn\\
	&\;  + \frac{2\gamma}{N}h^4a_{i-d} +  \frac{2\gamma}{N} c_i+ \gamma\frac{N-1}{N} e_{i-1}
}
We require that $\mu\leq \frac{\nu}{48\eta^2}$ such that
\eq{
	1-\mu\nu/3+4\mu^2\eta^2 \leq 1-\mu\nu/4
}
Rearranging terms:
\eq{
	&a_{i+1} + (\gamma-8\mu^2\delta^2 ) e_{i} \nn\\
	&\leq (1-\mu\nu/4) a_i + \mu dL'\sum_{j=i-d}^{i-1} b_{j+1}\nn\\
	&\;\;\;\;\;-\Big(\nu\mu - 8\mu^2L^2-\frac{2\gamma}{N}h^4\Big) a_{i-d}\nn\\
	&\;\; \;\;\;+\left(\frac{4\mu\delta^2}{\nu} + \frac{2\gamma}{N}\right) c_i + \gamma\frac{N-1}{N} e_{i-1} \nnb
	&= (1-\mu\nu/4) \left( a_i + \gamma\frac{N-1}{(1-\mu\nu/4)N}e_{i-1} \right) \nn\\
	&\;\;\;\;\;+\left(\frac{8\mu\delta^2}{\nu} + \frac{2\gamma}{N}\right)c_i -\Big(\nu\mu - 8\mu^2L^2-\frac{2\gamma}{N}h^4\Big) a_{i-d}\nn\\
	&\hspace{5mm} +\mu dL\sum_{j=i-d}^{i-1} b_{j+1}
}
We require that 
\eq{
	\gamma-8\mu^2\delta^2 =\,& \gamma\frac{N-1}{(1-\mu\nu/4)N} \nn\\
	\Longleftrightarrow\;\; [1 - \mu\nu N/4]\gamma =& 8\mu^2\delta^2N(1-\mu\nu/4)\nn\\
	\Longleftrightarrow\hspace{20mm}\gamma =& \frac{8\mu^2\delta^2N(1-\mu\nu/4)}{1 - \mu\nu N/4} = O(\mu^2)
}
We can further require that $\mu\le \frac{1}{2\nu N}$ so that $\gamma$ has the upper bound:
\eq{
	\gamma \leq 10\mu^2\delta^2N \leq 5\mu\frac{\delta^2}{\nu } \label{23g123e}
}
Let $s_{i+1} = a_{i+1}+(\gamma-8\mu^2\delta^2 ) e_{i}$, we have
\eq{
	s_{i+1} \leq (1-\mu\nu/3)s_i+ \mu d L\sum_{j=i-d}^{i-1} b_{j+1}+\left(\frac{8\mu\delta^2}{\nu} +\frac{2\gamma}{N}\right)c_i
}
where we discard the term about $a_{i-d}$, which requires:
\eq{
	&\hspace{-7mm}\nu\mu - 8\mu^2L^2-\frac{2\gamma}{N}h^4 \geq 0\nn\\
	\Longleftrightarrow\nu\mu \geq\,& 8\mu^2L^2+\frac{2\gamma}{N}h^4 \nn\\
	\stackrel{(\ref{23g123e})}{\Longleftarrow}\nu\mu \geq\,& 8\mu^2L^2 + {20\mu^2\delta^2}h^4\nn\\
	\Longleftrightarrow\,\;\mu \leq\,& \frac{\nu}{8L^2+20\delta^2 h^4}
}
After substituting \eqref{23g123e}, we have
\eq{
	s_{i+1} \leq& (1-\mu\nu/4)s_i+ \mu dL'\sum_{j=i-d}^{i-1} b_{j+1}+\left(\frac{8\mu\delta^2}{\nu} + \frac{10\mu\delta^2}{\nu N}\right)c_i\nn\\
	\leq& (1-\mu\nu/4)s_i + \mu dL' \sum_{j=i-d}^{i-1} \rho^{j+1} C_0+ \frac{18\mu\delta^2}{\nu}\rho^iC_3\nn\\
	\leq& (1-\mu\nu/4)s_i + \rho^{i} \left(\mu dL'\sum_{j=0}^{d-1} \rho^{-j} C_0 + \frac{18\mu\delta^2}{\nu} C_3\right)\nn\\
	\leq& \rho^{i+1} \!\!\left(\!s_0 + \frac{\mu}{\rho - 1 + \mu\nu/4}\left(dL'\sum_{j=0}^{d-1} \rho^{-j} C_0 \!+\! \frac{18\delta^2}{\nu} C_3\!\right)\!\!\right)\nn\\
	\define&\rho^{i+1} C_4
}
This also implies
\eq{
	a_{i} \leq \rho^i C_4
}
Revisiting \eqref{recursion.e}, we have
\eq{
	e_i\leq& \frac{N-1}{N} e_{i-1} + \rho^{i-d}\frac{h^4}{N}C_4\nn\\
	\leq&\rho^{i}\left(e_0 + \frac{1}{\rho -\frac{N-1}{N}}\frac{h^4}{\rho^{d+1}N}C_4\right)\nn\\
	\define&\rho^{i} C_5
}
Lastly, we substitute foregoing results into \eqref{recursion.b}:
\eq{
	b_{i+1} \leq& 8\mu^2 \delta^2 \rho^i C_3 + 8\mu^2\delta^2\rho^i C_5 + 8\mu^2L^2\rho^{i-d}C_4 + 4\mu^2\eta^2\rho^{i}C_4\nn\\
	=& \rho^{i}\mu^2\left(8\delta^2C_3 +8\delta^2C_5+\frac{C_4}{\rho^d}+4\eta^2C_4\right)
}
And we require the following to complete the mathematical induction:
\eq{\label{xhwej}
	\mu^2\left(6\delta^2C_3 + 6\delta^2C_5+\frac{C_4}{\rho^d}+4\eta^2C_4\right) < \rho C_0
}

In the following, we will show that $C_3$, $C_4$ and $C_5$ can be upper bounded by constants that are independent of step-size $\mu$. From \eqref{rho}, we have
\eq{
	\rho &\ge 1 - \frac{1 - \lambda^{J}}{2N}, \label{rho-lb-1}\\
	\rho &\ge 1 - \frac{\mu \nu}{4}. \label{rho-lb-2}
}
Moreover, since $\mu \le {1}/{2\nu N}$, from \eqref{rho-lb-2} we also have
\eq{
	\rho \ge 1 - \frac{\mu \nu}{4} \ge 1 - \frac{1}{8N}. \label{rho-lb-3}
}
With \eqref{rho-lb-1} and \eqref{rho-lb-2}, it holds that
\eq{
	\frac{1}{\rho - 1 + \frac{1}{2N}} &
	\overset{\eqref{rho-lb-1}}{\le} \frac{2N}{\lambda^J}, \label{rho-lbd-1}\\
	\frac{1}{\rho - 1 +\frac{1-\lambda^J}{N}} &\overset{\eqref{rho-lb-1}}{\le} \frac{2N}{1-\lambda^J}, \label{rho-lbd-2}\\
	\frac{\mu}{\rho - 1 + \mu\nu/4} &\overset{\eqref{rho-lb-2}}{\le} \frac{20}{\nu}, \label{rho-lbd-3}\\
	\frac{1}{N\rho -N+1} &\overset{\eqref{rho-lb-1}}{\le} \frac{2}{1+\lambda^J}. \label{rho-lbd-4}
}
Now we examine the upper bounds on $C_3$, $C_4$ and $C_5$. Note that
\eq{
	C_1 & = f_0 + \frac{1}{\rho - \frac{N-1/2}{N}}\frac{K^2+2K^2(N-1)^2}{N\rho^{d-1}}h^4C_0 \nnb
	& \overset{\eqref{rho-lbd-1}}{\le} f_0 + 2K^2 \frac{1 + 2(N-1)^2}{ \lambda^J\rho^{d-1}}h^4 C_0 \nnb
	& \overset{\eqref{rho-lb-3}}{\le} f_0 + 2K^2 \frac{1 + 2(N-1)^2}{\lambda^J(1-1/(8N))^{d-1}}h^4 C_0 = O(1),
}
and 
\eq{
	C_2 &= d_0 + \frac{1}{\rho - 1 + \frac{1-\lambda^J}{N} }\frac{\lambda^{2J}}{N(1-\lambda^J)}C_1 \nnb
	&\overset{\eqref{rho-lbd-2}}{\le} d_0 + \frac{2\lambda^{2J}}{(1-\lambda^J)^2} C_1 = O(1).
}
With $C_1$ and $C_2$, we have
\eq{
	C_3 = \lambda^J C_2 + \frac{\lambda^{2J}}{1-\lambda^J} C_1 \define C_3'= O(1)
}
Next we examine $C_4$:
\eq{
	C_4 &= s_0 + \frac{\mu}{\rho - 1 + \mu\nu/4}\left(dL'\sum_{j=0}^{d-1} \rho^{-j} C_0 + \frac{18\delta^2}{\nu} C_3\right) \nnb
	&\overset{\eqref{rho-lbd-3}}{\le} s_0 + \frac{20}{\nu}\left(\frac{d^2L'}{(1-\frac{1}{8N})^d} C_0 + \frac{18\delta^2}{\nu} C_3\right) \define C_4'= O(1).
}
For term $C_5$:
\eq{
	C_5 &= e_0 + \frac{h^4}{\rho^{d+1}(N\rho -N+1)}C_4 \nnb
	&\overset{\eqref{rho-lbd-4}}{\le} e_0 + \frac{2h^4}{\rho^{d+1}(1+\lambda^J)}C_4 \nnb
	&\le e_0 + \frac{2h^4}{(1-1/(8N))^{d+1}(1+\lambda^J)}C_4 \define C_5' = O(1).
}
Note that all constants $C_3'$, $C_4'$ and $C_5'$ are independent of $\mu$. Finally, note that 
\eq{
	&\hspace{-4mm}6\delta^2C_3 + 6\delta^2C_5+\frac{C_4}{\rho^d} + 4\eta^2C_4\nn\\ 
	\le &\ 6\delta^2C_3' + 6\delta^2C_5'+\frac{C_4'}{\rho^d}+ 4\eta^2C'_4 \nnb
	\le &\ 6\delta^2C_3' + 6\delta^2C_5'+\frac{C_4'}{(1-1/(8N))^d}+ 4\eta^2C'_4\nnb
	\define& B,
}
and $B$ is independent of step-size $\mu$. Also, we have have
\eq{
	\rho C_0 \ge \left(1- \frac{1}{8N}\right) C_0.
}
To prove \eqref{xhwej}, it is enough to choose $\mu$ such that
\eq{
	\mu^2 B \le \left(1- \frac{1}{8N}\right) C_0 \Longleftrightarrow \mu \le \sqrt{\left(1- \frac{1}{8N}\right)\frac{C_0}{B}}.
}
Combining with ${1}/{2\nu N}$, we can set $\mu$ as
\eq{
	\mu \le \min \left\{\sqrt{\left(1- \frac{1}{8N}\right)\frac{C_0}{B}}, \frac{1}{2\nu N} \right\}
}

\bibliographystyle{IEEEbib}
\bibliography{distributed_feature}
\end{document}